\newif\ifdoublecol

\doublecoltrue
\documentclass[lettersize,journal]{IEEEtran}
\IEEEoverridecommandlockouts

\usepackage{etex} % Solves: "no room for new dimensions"

\ifCLASSINFOpdf
  \usepackage[pdftex]{graphicx}
  
  \usepackage{epstopdf}
\else
  \usepackage[dvips]{graphicx}
  \graphicspath{{../eps/}}
  \DeclareGraphicsExtensions{.eps}
\fi
\usepackage{cite}

\usepackage[cmex10]{amsmath}
\usepackage{xcolor,colortbl}
\definecolor{col1}{HTML}{3891A6}
\definecolor{col2}{HTML}{EF5B5B}
\definecolor{col3}{HTML}{3DDC97}

\usepackage{amsmath,amsthm,relsize}
\usepackage{amsfonts, amssymb, cuted}
\usepackage{cleveref}
\usepackage{mathtools}
\usepackage{algorithm}
\usepackage[noend]{algpseudocode}
\usepackage{cite}
\usepackage{soul}
\usepackage{graphicx}
\usepackage{tikz}
\usepackage{pgfplotstable}
\usepackage{pgfplots}
\usepackage{nicefrac}
\usepackage{enumitem}
\usepackage{support-caption}
\usepackage{subcaption}
\usepackage[font=footnotesize]{caption}
\usepackage{multirow}
\pgfplotsset{compat=1.15}
\usepackage{relsize}
\usetikzlibrary{patterns}
\usepackage{acro}
\usepackage{pifont}

\newtheorem{remark}{Remark}
\usepackage{todonotes}
\newtheorem{definition}{Definition}
\usepackage{mathtools}

\usepackage{color}
\usepackage{colortbl}
\usepackage{multirow}

\newlength{\Oldarrayrulewidth}
\definecolor{intnull}{RGB}{213,229,255}
\definecolor{inteins}{RGB}{128,179,255}
\definecolor{intzwei}{RGB}{42,127,255}
\definecolor{intdrei}{RGB}{0,85,212}
\definecolor{intvier}{RGB}{0,51,128}
\definecolor{intfunf}{RGB}{0,34,85}
\usetikzlibrary{decorations.pathreplacing}

\newtheorem{lemma}{Lemma}

\newcommand{\herm}{^{\mbox{\scriptsize H}}}

\newcommand{\vbar}{\raisebox{.17ex}{\rule{.04em}{1.35ex}}}
\newcommand{\vbarind}{\raisebox{.01ex}{\rule{.04em}{1.1ex}}}
\newcommand{\R}{\ifmmode{\rm I}\hspace{-.2em}{\rm R} \else ${\rm I}\hspace{-.2em}{\rm R}$ \fi}
\newcommand{\T}{\ifmmode{\rm I}\hspace{-.2em}{\rm T} \else ${\rm I}\hspace{-.2em}{\rm T}$ \fi}
\newcommand{\N}{\ifmmode{\rm I}\hspace{-.2em}{\rm N} \else \mbox{${\rm I}\hspace{-.2em}{\rm N}$} \fi}
\newcommand{\B}{\ifmmode{\rm I}\hspace{-.2em}{\rm B} \else \mbox{${\rm I}\hspace{-.2em}{\rm B}$} \fi}
\newcommand{\Hil}{\ifmmode{\rm I}\hspace{-.2em}{\rm H} \else \mbox{${\rm I}\hspace{-.2em}{\rm H}$} \fi}
\newcommand{\C}{\ifmmode\hspace{.2em}\vbar\hspace{-.31em}{\rm C} \else \mbox{$\hspace{.2em}\vbar\hspace{-.31em}{\rm C}$} \fi}
\newcommand{\Cind}{\ifmmode\hspace{.2em}\vbarind\hspace{-.25em}{\rm C} \else \mbox{$\hspace{.2em}\vbarind\hspace{-.25em}{\rm C}$} \fi}
\newcommand{\Q}{\ifmmode\hspace{.2em}\vbar\hspace{-.31em}{\rm Q} \else \mbox{$\hspace{.2em}\vbar\hspace{-.31em}{\rm Q}$} \fi}
\newcommand{\Z}{\ifmmode{\rm Z}\hspace{-.28em}{\rm Z} \else ${\rm Z}\hspace{-.28em}{\rm Z}$ \fi}

\newtheorem{exmp}{Example}%[section]
\theoremstyle{definition}

\newcommand{\CA}[0]{{\mathcal{A}}}
\newcommand{\CB}[0]{{\mathcal{B}}}

\newcommand{\CF}[0]{{\mathcal{F}}}

\newcommand{\CI}[0]{{\mathcal{I}}}

\newcommand{\CK}[0]{{\mathcal{K}}}

\newcommand{\CM}[0]{{\mathcal{M}}}
\newcommand{\CN}[0]{{\mathcal{N}}}

\newcommand{\CQ}[0]{{\mathcal{Q}}}

\newcommand{\CS}[0]{{\mathcal{S}}}

\newcommand{\CU}[0]{{\mathcal{U}}}

\newcommand{\Bb}[0]{{\mathbf{b}}}

\newcommand{\Bd}[0]{{\mathbf{d}}}

\newcommand{\Bh}[0]{{\mathbf{h}}}

\newcommand{\Bq}[0]{{\mathbf{q}}}

\newcommand{\Bv}[0]{{\mathbf{v}}}

\newcommand{\Bx}[0]{{\mathbf{x}}}
\newcommand{\By}[0]{{\mathbf{y}}}

\newcommand{\BA}[0]{{\mathbf{A}}}

\newcommand{\BG}[0]{{\mathbf{G}}}

\newcommand{\BI}[0]{{\mathbf{I}}}

\newcommand{\BP}[0]{{\mathbf{P}}}
\newcommand{\BQ}[0]{{\mathbf{Q}}}

\DeclareAcronym{ADMM}{
    short = ADMM,
    long = alternating direction method of multipliers,
    list = Alternating Direction Method of Multipliers,
    tag = abbrev
}

\DeclareAcronym{AoA}{
    short = AoA,
    long = angle-of-arrival,
    list = Angle-of-Arrival,
    tag = abbrev
}

\DeclareAcronym{SISO}{
    short = SISO,
    long = single-input single-output,
    list = single-input single-output,
    tag = abbrev
}

\DeclareAcronym{MRT}{
    short = MRT,
    long = maximum ratio transmitter,
    list = maximum ratio transmitter,
    tag = abbrev
}

\DeclareAcronym{PDA}{
    short = PDA,
    long = placement delivery array,
    list = placement delivery array,
    tag = abbrev
}

\DeclareAcronym{EE}{
    short = EE,
    long = energy efficiency,
    list = energy efficiency,
    tag = abbrev
}

\DeclareAcronym{MDS}{
    short = MDS,
    long = maximum distance separation,
    list = maximum distance separation,
    tag = abbrev
}

\DeclareAcronym{SIC}{
    short = SIC,
    long = successive-interference-cancellation,
    list = successive-interference-cancellation,
    tag = abbrev
}

\DeclareAcronym{MAC}{
    short = MAC,
    long = multiple-access-channel,
    list = multiple-access-channel,
    tag = abbrev
}

\DeclareAcronym{AoD}{
    short = AoD,
    long = angle-of-departure,
    list = Angle-of-Departure,
    tag = abbrev
}

\DeclareAcronym{BB}{
    short = BB,
    long = base band,
    list = Base Band,
    tag = abbrev
}

\DeclareAcronym{BC}{
    short = BC,
    long = broadcast channel,
    list = Broadcast Channel,
    tag = abbrev
}

\DeclareAcronym{BS}{
    short = BS,
    long = base station,
    list = Base Station,
    tag = abbrev
}

\DeclareAcronym{BR}{
    short = BR,
    long = best response,
    list = Best Response, 
    tag = abbrev
}

\DeclareAcronym{CB}{
    short = CB,
    long = coordinated beamforming,
    list = Coordinated Beamforming,
    tag = abbrev
}

\DeclareAcronym{CC}{
    short = CC,
    long = coded caching,
    list = Coded Caching,
    tag = abbrev
}

\DeclareAcronym{CE}{
    short = CE,
    long = channel estimation,
    list = Channel Estimation,
    tag = abbrev
}

\DeclareAcronym{CoMP}{
    short = CoMP,
    long = coordinated multi-point transmission,
    list = Coordinated Multi-Point Transmission,
    tag = abbrev
}

\DeclareAcronym{CRAN}{
    short = C-RAN,
    long = cloud radio access network,
    list = Cloud Radio Access Network,
    tag = abbrev
}

\DeclareAcronym{CSE}{
    short = CSE,
    long = channel specific estimation,
    list = Channel Specific Estimation,
    tag = abbrev
}

\DeclareAcronym{CSI}{
    short = CSI,
    long = channel state information,
    list = Channel State Information,
    tag = abbrev
}

\DeclareAcronym{CSIT}{
    short = CSIT,
    long = channel state information at the transmitter,
    list = Channel State Information at the Transmitter,
    tag = abbrev
}

\DeclareAcronym{CU}{
    short = CU,
    long = central unit,
    list = Central Unit,
    tag = abbrev
}

\DeclareAcronym{D2D}{
    short = D2D,
    long = device-to-device,
    list = Device-to-Device,
    tag = abbrev
}

\DeclareAcronym{DE-ADMM}{
    short = DE-ADMM,
    long = direct estimation with alternating direction method of multipliers,
    list = Direct Estimation with Alternating Direction Method of Multipliers,
    tag = abbrev
}

\DeclareAcronym{DE-BR}{
    short = DE-BR,
    long = direct estimation with best response,
    list = Direct Estimation with Best Response,
    tag = abbrev
}

\DeclareAcronym{DE-SG}{
    short = DE-SG,
    long = direct estimation with stochastic gradient,
    list = Direct Estimation with Stochastic Gradient,
    tag = abbrev
}

\DeclareAcronym{DFT}{
	short = DFT,
	long = discrete fourier transform,
	list = Discrete Fourier Transform,
	tag = abbrev
}

\DeclareAcronym{DoF}{
    short = DoF,
    long = degrees of freedom,
    list = Degrees of Freedom,
    tag = abbrev
}

\DeclareAcronym{DL}{
    short = DL,
    long = downlink,
    list = Downlink,
    tag = abbrev
}

\DeclareAcronym{GD}{
	short = GD, 
	long = gradient descent,
	list = Gradeitn Descent,
	tag = abbrev
}

\DeclareAcronym{IBC}{
    short = IBC,
    long = interfering broadcast channel,
    list = Interfering Broadcast Channel,
    tag = abbrev
}

\DeclareAcronym{i.i.d.}{
    short = i.i.d.,
    long = independent and identically distributed,
    list = Independent and Identically Distributed,
    tag = abbrev
}

\DeclareAcronym{JP}{
    short = JP,
    long = joint processing,
    list = Joint Processing,
    tag = abbrev
}

\DeclareAcronym{KKT}{
    short = KKT,
    long = Karush-Kuhn-Tucker,
    tag = abbrev
}

\DeclareAcronym{LOS}{
	short = LOS,
	long = line-of-sight,
	list = Line-of-Sight,
	tag = abbrev
}

\DeclareAcronym{LS}{
    short = LS,
    long = least squares,
    list = Least Squares,
    tag = abbrev
}

\DeclareAcronym{LTE}{
    short = LTE,
    long = Long Term Evolution,
    tag = abbrev
}

\DeclareAcronym{LTE-A}{
    short = LTE-A,
    long = Long Term Evolution Advanced,
    tag = abbrev
}

\DeclareAcronym{MIMO}{
    short = MIMO,
    long = multiple-input multiple-output,
    list = Multiple-Input Multiple-Output,
    tag = abbrev
}

\DeclareAcronym{MISO}{
    short = MISO,
    long = multiple-input single-output,
    list = Multiple-Input Single-Output,
    tag = abbrev
}

\DeclareAcronym{MSE}{
    short = MSE,
    long = mean-squared error,
    list = Mean-Squared Error,
    tag = abbrev
}

\DeclareAcronym{MMSE}{
    short = MMSE,
    long = minimum mean-squared error,
    list = Minimum Mean-Squared Error,
    tag = abbrev
}

\DeclareAcronym{mmWave}{
	short = mmWave,
	long = millimeter wave,
	list = Millimeter Wave,
	tag = abbrev
}

\DeclareAcronym{MU-MIMO}{
    short = MU-MIMO,
    long = multi-user \ac{MIMO},
    list = Multi-User \ac{MIMO},
    tag = abbrev
}

\DeclareAcronym{OTA}{
    short = OTA,
    long = over-the-air,
    list = Over-the-Air,
    tag = abbrev
}

\DeclareAcronym{PSD}{
    short = PSD,
    long = positive semidefinite,
    list = Positive Semidefinite,
    tag = abbrev
}

\DeclareAcronym{QoS}{
	short = QoS,
	long = quality of service,
	list = Quality of Service,
	tag = abbrev
}

\DeclareAcronym{RCP}{
	short = RCP,
	long = remote central processor,
	list = Remote Central Processor,
	tag = abbrev
}

\DeclareAcronym{RRH}{
    short = RRH,
    long = remote radio head,
    list = Remote Radio Head,
    tag = abbrev
}

\DeclareAcronym{RSSI}{
    short = RSSI,
    long = received signal strength indicator,
    list = Received Signal Strength Indicator,
    tag = abbrev
}

\DeclareAcronym{RX}{
	short = RX,
	long = receiver,
	list = Receiver,
	tag = abbrev
}

\DeclareAcronym{SCA}{
    short = SCA,
    long = successive convex approximation,
    list = Successive Convex Approximation,
    tag = abbrev
}

\DeclareAcronym{SG}{
    short = SG,
    long = stochastic gradient,
    list = Stochastic Gradient,
    tag = abbrev
}

\DeclareAcronym{SNR}{
    short = SNR,
    long = signal-to-noise ratio,
    list = Signal-to-Noise Ratio,
    tag = abbrev
}

\DeclareAcronym{SINR}{
    short = SINR,
    long = signal-to-interference-plus-noise ratio,
    list = Signal-to-Interference-plus-Noise Ratio,
    tag = abbrev
}

\DeclareAcronym{SOCP}{
	short = SOCP, 
	long = second order cone program,
	list = Second Order Cone Program,
	tag = abbrev
}

\DeclareAcronym{SSE}{
    short = SSE,
    long = stream specific estimation,
    list = Stream Specific Estimation,
    tag = abbrev
}

\DeclareAcronym{SVD}{
	short = SVD,
	long = singular value decomposition,
	list = Singular Value Decomposition,
	tag = abbrev
}

\DeclareAcronym{TDD}{
	short = TDD,
	long = time division duplex,
	list = Time Division Duplex,
	tag = abbrev
}

\DeclareAcronym{TX}{
	short = TX,
	long = transmitter,
	list = Transmitter,
	tag = abbrev
}

\DeclareAcronym{UE}{
    short = UE,
    long = user equipment,
    list = User Equipment,
    tag = abbrev
}

\DeclareAcronym{UL}{
    short = UL,
    long = uplink,
    list = Uplink,
    tag = abbrev
}

\DeclareAcronym{ULA}{
	short = ULA,
	long = uniform linear array,
	list = Uniform Linear Array,
	tag = abbrev
}

\DeclareAcronym{UPA}{
    short = UPA,
    long = uniform planar array,
    list = Uniform Planar Array,
    tag = abbrev
}

\DeclareAcronym{WMMSE}{
    short = WMMSE,
    long = weighted minimum mean-squared error,
    list = Weighted Minimum Mean-Squared Error,
    tag = abbrev
}

\DeclareAcronym{WMSEMin}{
    short = WMSEMin,
    long = weighted sum \ac{MSE} minimization,
    list = Weighted sum \ac{MSE} Minimization,
    tag = abbrev
}

\DeclareAcronym{WBAN}{
	short = WBAN,
	long = wireless body area network,
	list = Wireless Body Area Network,
	tag = abbrev
}

\DeclareAcronym{WSRMax}{
    short = WSRMax,
    long = weighted sum rate maximization,
    list = Weighted Sum Rate Maximization,
    tag = abbrev
}

\begin{document}

\title{Low-Complexity Multi-Antenna Coded Caching Using Location-Aware Placement Delivery Arrays}

\author{\IEEEauthorblockN{Hamidreza Bakhshzad Mahmoodi,~\IEEEmembership{Student~Member,~IEEE,}  MohammadJavad Salehi,~\IEEEmembership{Member,~IEEE,} and Antti T\"olli,~\IEEEmembership{Senior~Member,~IEEE}}
\thanks{This work was supported by the Academy of Finland under grants no. 319059 (Coded Collaborative Caching for Wireless Energy Efficiency) and 346208 (6G Flagship program). This article was presented in parts at the IEEE International Conference on Communications (ICC) 2022, Seoul, Korea. \textit{(Corresponding~author: Hamidreza Bakhshzad Mahmoodi). } Hamidreza Bakhshzad Mahmoodi, MohammadJavad Salehi and Antti T\"olli are with the Centre for Wireless Communications, University of Oulu, FIN-90014 Oulu, Finland. (e-mail:First-Name.Last-Name@oulu.fi)}
}

\maketitle

%!TeX root = d2d-cc.tex

\begin{abstract}
A location-aware multi-antenna coded caching scheme is proposed for applications with location-dependent data requests, such as wireless immersive experience, where users are immersed in a three-dimensional virtual world. The wireless connectivity conditions vary as the users move within the application area motivating the use of a non-uniform cache memory allocation process to avoid excessive delivery time for users located in wireless bottleneck areas. To this end, a location-aware placement and delivery array (LAPDA) is designed for cache-aided multiantenna data delivery with a fast converging, iterative linear beamforming process. The underlying weighted max-min transmit precoder design enables the proposed scheme to serve users in poor connectivity areas with smaller amounts of data while simultaneously delivering larger amounts to other users. Our new scheme is suitable for large networks due to its linear transceiver structure and it is not constrained by the number of users, cache size, or the number of antennas at the transmitter, unlike the existing schemes. Despite non-uniform cache placement, the proposed scheme still achieves a significant degree of coded caching gain that is additive to the multiplexing gain and greatly outperforms the conventional symmetric CC schemes in terms of both average and 95-percentile delivery time.


\end{abstract}

\section{Introduction}
Mobile data traffic is exponentially growing, and this trend will continue as the market is constantly inundated with technologies and devices that support new data-intensive applications in different forms and capabilities~\cite{cisco2020}. Wireless eyewear devices, for example, enable data-intensive mobile extended reality (XR) applications~\cite{salehi2022enhancing, flashback_2016_VR_static_dynamic_support}, which are also subject to strict quality of service (QoS) requirements such as low latency ($ < 10$ ms) and high data rate transmission ($6.37 - 95.55$ Gbps)~\cite{Nokia-immersive,6G_white_paper_2020,bastug2017toward, taleb2022_towards_XR_vision, walid_sad_bennis_VR_XR_2022,wireless_virtual_reality_TCOM2018}. This differs greatly from conventional ultra-low latency and low-rate requirements for internet-of-things applications~\cite{URLLC_ICC_2017}. The low latency, along with the high delivery rates, require more sophisticated transmission methods than those offered by current wireless network standards~\cite{bastug2017toward, taleb2022_towards_XR_vision, walid_sad_bennis_VR_XR_2022}. Therefore, to meet the requirements of future wireless XR applications, new delivery 
schemes with higher bandwidth efficiency are needed. 

One possible option, given that
upcoming mobile broadband applications rely heavily on asynchronous content reuse~\cite{proactive_caching_Cm2016}, is to utilize proactive caching at the end-users to relieve network congestion and bandwidth consumption during peak times~\cite{role_of_caching_in_future_wireless_caire_JSAC2018}. In this regard, various studies have explored proactive caching in \ac{SISO} configurations, demonstrating its benefits for meeting XR application requirements~\cite{sun2019communications,yang2018communication,sun2020bandwidth,dang2019joint}. Specifically, with the available memory at the end users, the whole or part of the requested content can be cached beforehand and rendered by the end user at the request time. This results in significant bandwidth and delay-reduction gains and alleviates the traffic burden over the wireless network~\cite{sun2019communications,yang2018communication,sun2020bandwidth,dang2019joint}.  

Unlike conventional caching schemes that rely on the available memory of each user (see, e.g.,~\cite{proactive_caching_Cm2016, role_of_caching_in_future_wireless_caire_JSAC2018, sun2019communications,yang2018communication,sun2020bandwidth,dang2019joint, bastug2014living,yang2015analysis}), the \ac{CC} scheme originally introduced in~\cite{MaddahAli-2014} benefits from the aggregated memory throughout the network. In fact, it enjoys a so-called \textit{global caching} gain, available through careful cache placement and multicast transmissions that results in improved overall performance compared to traditional schemes~\cite{proactive_caching_Cm2016, role_of_caching_in_future_wireless_caire_JSAC2018, sun2019communications,yang2018communication,sun2020bandwidth,dang2019joint, bastug2014living,yang2015analysis}. As such, the transmission bandwidth for delay-constrained XR applications has been effectively reduced in \ac{SISO} setups by leveraging coded cache placement and mobile device computing capabilities~\cite{CC_edge_computing_for_VR_twc2021}. The \ac{CC} scheme is especially advantageous for large networks as the achievable global caching gain scales linearly with the number of users in the network. This makes it ideal for collaborative XR scenarios where a group of users is served simultaneously within a confined environment, with each user's individual actions impacting the results perceived by all users (c.f.,~\cite{Mahmoodi_immersive_isit2021, mahmoodi2022asymmetric,mahmoodi_ICC2022_nonsym, mahmoodi2022asymmetric_arxiv_TWC2023, salehi2022enhancing}). In this regard, a \emph{location-dependent} CC-based cache placement and delivery scheme, originally designed for \ac{SISO} setups, has been proposed in~\cite{Mahmoodi_immersive_isit2021}.

In addition, the \ac{CC} scheme's ability to combine global caching and spatial multiplexing gains is another critical feature~\cite{pooya-cc-physical-2019-journal}. This is particularly appealing given that multi-antenna connectivity will be a crucial feature of upcoming communication systems~\cite{6G_white_paper_2020}. Thus, the SISO setup in~\cite{Mahmoodi_immersive_isit2021} has been extended to a \ac{MISO} setup in~\cite{ mahmoodi2022asymmetric,mahmoodi_ICC2022_nonsym, mahmoodi2022asymmetric_arxiv_TWC2023}, to benefit from spatial multiplexing and global caching gains simultaneously. In this paper, we intend to overcome some of the practical limitations of our earlier schemes in~\cite{ mahmoodi2022asymmetric,mahmoodi_ICC2022_nonsym, mahmoodi2022asymmetric_arxiv_TWC2023}. Notably, we propose a new location-dependent \ac{CC} scheme, 
%agnostic to the baseline \ac{CC} scheme used for data placement and codeword creation.
%The new scheme is 
suitable for large networks due to its linear transceiver structure, and
%it is 
not constrained by the number of users, cache size, or the number of antennas at the transmitter, unlike the existing schemes.

%This is unlike~\cite{mahmoodi_ICC2022_nonsym} that required the spatial multiplexing gain to be greater than or equal to the global caching gain. Additionally, we propose a simple transceiver design that makes the proposed scheme applicable to much larger networks compared to~\cite{mahmoodi2022asymmetric,mahmoodi2022asymmetric_arxiv_TWC2023}.

%by using a simpler transceiver design which allows it to be applied to much larger networks. In addition,  the proposed scheme
%is agnostic to the baseline \ac{CC} scheme used for data placement and codeword creation; hence, it works for every given network setup and resolves the applicability bottleneck of the location-dependent \ac{CC} scheme of~\cite{mahmoodi_ICC2022_nonsym}, which worked only when the spatial multiplexing gain at the transmitter was not smaller than the achievable \ac{CC} gain.
%does not require the multiplexing gain to be higher than the \ac{CC} gain, unlike~\cite{mahmoodi_ICC2022_nonsym}.
%%%%%%%%%%%%%%%%%%%%%%%%%%%%%%%%%%%%%%%%%%%%%%%%%%%%

\subsection{Literature review}

\noindent\textbf{Coded caching.} The original \ac{CC} scheme in~\cite{MaddahAli-2014} was intended for \ac{SISO} setups with an error-free shared link. This work was later extended to more practical scenarios, including multi-server~\cite{Shariatpanahi2016} and \ac{MISO}~\cite{pooya-cc-physical-2019-journal,yu2017exact,tolli2017multi} setups. The early high \ac{SNR} analysis in~\cite{Shariatpanahi2016,pooya-cc-physical-2019-journal,yu2017exact} proved that the so-called \ac{DoF} achieved by the \ac{MISO}-\ac{CC} scheme is optimal under uncoded cache placement and single-shot data delivery. Later, the analysis in~\cite{tolli2017multi} showed that an optimized multi-antenna precoder design is necessary for the \ac{CC} scheme to perform well also in the low-SNR regime. Soon, device-to-device (D2D) \ac{CC} schemes were proposed (e.g.,~\cite{Ji2016,D2D-CC-Optload-memtradeof-caire-2019,D2D_CC_mahmoodi_2022}) to increase the network throughput. %In this regard, an infrastructure-less network (i.e., with D2D links only) was studied in~\cite{Ji2016}, where CC-based cache placement and delivery schemes were designed to achieve an optimal transmission rate~\cite{D2D-CC-Optload-memtradeof-caire-2019}. Later on, this system model was extended in~\cite{D2D_CC_mahmoodi_2022} to a general framework where D2D transmissions complemented downlink transmission.

\noindent\textbf{Bit- and signal-level \ac{CC}.} Despite exciting theoretical gains, various practical issues have restricted the real-world implementation of \ac{CC} schemes. One prominent problem is the exponentially growing file-division requirement (w.r.t the network size), known as the \textit{subpacketization} bottleneck. To address this issue, a combinatorial subfile assignment based on placement delivery arrays (PDA) was proposed in~\cite{PDA_first_2017}.
%, which allowed the placement and delivery scheme design to be translated to a matrix form.
The PDA structure provides a set of conditions (reviewed for \ac{MISO} systems in Section~\ref{sec:cache_placement}) that allows a given matrix to be used for both content placement and delivery of a \ac{CC} scheme, thus translating the subpacketization reduction problem to finding a small-dimension matrix satisfying PDA conditions.
%The PDA structure provided conditions on a given matrix could be used for both content placement and delivery of a \ac{CC} scheme, thus translating the subpacketization reduction problem to finding 
%under well-defined conditions. 
Interestingly, authors in~\cite{PDA_first_2017} demonstrated that all the schemes in~\cite{MaddahAli-2014,Shariatpanahi2016,pooya-cc-physical-2019-journal,yu2017exact,tolli2017multi,Ji2016,D2D-CC-Optload-memtradeof-caire-2019,D2D_CC_mahmoodi_2022} could also be presented as PDAs. Motivated by the generalized framework in~\cite{PDA_first_2017}, various PDA-based \ac{CC} schemes were later proposed for different settings, aiming for reduced subpacketization~\cite{cheng_generalized_PDA_2019,DPDA2019,cheng_centralized_PDA_framework_2021}.

A major breakthrough in subpacketization reduction was achieved with the introduction of \emph{signal-level} \ac{CC} schemes in~\cite{lampiris2018adding}. In contrast to \emph{bit-level} \ac{CC} schemes~\cite{MaddahAli-2014,Shariatpanahi2016,pooya-cc-physical-2019-journal,yu2017exact,tolli2017multi,Ji2016,D2D-CC-Optload-memtradeof-caire-2019,D2D_CC_mahmoodi_2022}, where file fragments intended to different users are combined/separated using bit-wise XOR operations in the finite field, signal-level \ac{CC} schemes rely on the superposition of all precoded data terms in the signal domain and the regeneration and cancellation of the unwanted parts at the physical layer of each receiver 
%multiplexing/demultiplexing 
(see~\cite{salehi2022enhancing} for a more detailed explanation). As a result, the design flexibility is greatly increased compared with bit-level schemes, enabling the subpacketization requirement of \ac{MISO}-\ac{CC} setups to be even smaller than their comparable \ac{SISO}-\ac{CC} settings~\cite{lampiris2018adding}. 
%(it was even shown in~\cite{salehi2020lowcomplexity} that the exponential subpacketization growth of \ac{MISO}-\ac{CC} schemes could be replaced by a linear one under well-defined conditions). 
The signal-level scheme of~\cite{lampiris2018adding} was then extended to centralized~\cite{Shahred_cache_Emanuele_2020} and decentralized~\cite{Shared_cache_decentrlized_2021} shared-cache scenarios where a limited number of cache-enabled helper nodes serve a group of cache-less users. The applicability of signal-level \ac{CC} schemes was later extended, e.g., to \ac{MIMO} setups~\cite{salehi2021MIMO,salehi2023multicast}, and to dynamic networks wherein users may freely enter/depart the network at will~\cite{abolpour2022coded,abolpour2023cache,salehi2021lowsubpacketization}. 
Finally, to make the design of signal-level \ac{CC} schemes more systematic, an enhanced PDA framework, called multi-antenna placement and delivery arrays (MLPDA), was proposed in~\cite{MLPDA_ISIT2022}. Of course, signal-level \ac{CC} schemes also suffer from drawbacks such as inferior finite-SNR performance compared with bit-level \ac{CC} schemes~\cite{salehi2019subpacketization,salehi2022multi} and the requirement to regenerate and remove the interference in the physical layer~\cite{salehi2022enhancing}. However, the remarkable flexibility of signal-level approaches continues to inspire ongoing research endeavors aimed at utilizing them to overcome different implementation challenges encountered in \ac{CC} schemes.

\noindent\textbf{The near-far effect.} Another crucial problem of conventional  \ac{CC} schemes is the \textit{near-far} issue, which affects content delivery applications (e.g.,~\cite{Shariatpanahi2016,pooya-cc-physical-2019-journal,yu2017exact,tolli2017multi,Ji2016,D2D-CC-Optload-memtradeof-caire-2019,D2D_CC_mahmoodi_2022, PDA_first_2017,cheng_generalized_PDA_2019,DPDA2019,cheng_centralized_PDA_framework_2021}) in general and XR applications in particular. In \ac{CC} schemes, a common multicast message is transmitted to serve several users at a time, and all these users must be able to decode the message simultaneously. As a result, the achievable rate is always limited by the user(s) with the worst channel condition. Studies on \ac{SISO}-\ac{CC} networks have shown that the practical gains of \ac{CC} schemes could entirely vanish at the low-\ac{SNR} region due to the near-far issue~\cite{zhao_petros_MU_MISO_near_far_issue_ITW2021}. 
%The near-far issue also exists in signal-level schemes~\cite{lampiris2018adding,salehi2020lowcomplexity,Shahred_cache_Emanuele_2020,Shared_cache_decentrlized_2021,MLPDA_ISIT2022} as the delivery time is still limited by the user(s) with the poorest channel quality.
To address this issue, a congestion control technique is proposed in~\cite{destounis2020adaptive} with the intention of avoiding serving users that experience adverse channel conditions altogether. Similar scheduling approaches are also proposed in~\cite{Coded_caching_for_stochastic_wireless_network_TWC2021,liu_joint_power_energi_cc_TWC2021}, where joint queue minimization and packet control, as well as power minimization and scheduling, are considered for delay-constrained \ac{CC} applications. %Unlike~\cite{destounis2020adaptive,Coded_caching_for_stochastic_wireless_network_TWC2021,liu_joint_power_energi_cc_TWC2021}, in which users with bad wireless connectivity are ignored for data delivery, 
In another work~\cite{salehi2020coded}, a \ac{CC} scheme with partial codewords (i.e., with a smaller number of data terms in the codeword compared to the baseline \ac{CC} scheme of~\cite{maddah2014fundamental}) is introduced to adjust the user-specific QoE based on their current channel conditions.%serve users with poor channel conditions with a lower quality of experience (QoE) than the well-conditioned ones.

\noindent\textbf{Multi-rate transmission in \ac{CC}.}
%Using a different perspective, it was discussed in~\cite{ozfatura_mobility_awaire_RCOM2020} that as long as user mobility patterns were known at the server, different cache profiles could be assigned to multiple cache-enabled helper nodes scattered throughout the environment to improve the \ac{CC} transmission rate. Moreover, guiding users towards locations with preferable conditions in an immersive XR application using learning-based techniques was considered in~\cite{reinforcement_for_immersive_elsevier2017}, and an order-optimal location-based coded cache placement was proposed in~\cite{caire_optimal_location_dependent_arxiv2021} to assign different cache profiles to cache-enabled transmitters located in distinct locations. 
In~\cite{Shariatpanahi2016,pooya-cc-physical-2019-journal,yu2017exact,tolli2017multi,Ji2016,D2D-CC-Optload-memtradeof-caire-2019, D2D_CC_mahmoodi_2022, cheng_generalized_PDA_2019, DPDA2019, cheng_centralized_PDA_framework_2021, PDA_first_2017, MLPDA_ISIT2022, lampiris2018adding,salehi2020lowcomplexity,Shahred_cache_Emanuele_2020,Shared_cache_decentrlized_2021}, equal-sized data chunks are combined to form a common message.
%, making the symmetric rate necessary to minimize the overall transmission time and limiting the performance to the users with the worst channel condition. 
In contrast, in~\cite{tang2017coded}, data terms with different sizes are combined via nested code modulation (NCM), creating codewords that serve every user in the multicasting group at a different rate. Similarly, combining the shared-cache idea of~\cite{lampiris2018adding} with the NCM of~\cite{tang2017coded}, the near-far problem was mitigated in~\cite{zhao_wireless_CC_ring_area_ISIT2021,zhao_petros_near_far_Nakagami_asilomar2021}. The proposed system model in~\cite{tang2017coded} assumed fixed link capacities, particularly tailored for backhaul networks. Motivated by the results in~\cite{tang2017coded},
a location-dependent \ac{CC} scheme was proposed in~\cite{Mahmoodi_immersive_isit2021} for networks with variable link capacities,
%authors in~\cite{Mahmoodi_immersive_isit2021} proposed a location-dependent \ac{CC} network with variable link capacity, 
particularly applicable for future wireless XR applications. Later, the \ac{SISO} setting in~\cite{Mahmoodi_immersive_isit2021} was extended to a location-dependent \ac{MISO} setup in~\cite{mahmoodi_ICC2022_nonsym,mahmoodi2022asymmetric,mahmoodi2022asymmetric_arxiv_TWC2023}. Specifically, while the schemes in~\cite{mahmoodi2022asymmetric,mahmoodi2022asymmetric_arxiv_TWC2023} benefit from the NCM for data delivery, the scheme in~\cite{mahmoodi_ICC2022_nonsym} uses a modified version of the signal-level scheme in~\cite{salehi2020lowcomplexity} to support multi-rate transmission. Nevertheless, the schemes proposed in~\cite{mahmoodi2022asymmetric,mahmoodi2022asymmetric_arxiv_TWC2023} do not scale well with the increasing number of users. This is attributed to the exponential increase in the number of variables and constraints in the transmit precoder design optimization problem and the complex receiver structure. %due to exponential numbers of conditions and variables in the transmitter optimization problem and the complex receiver structure. 
Similarly, the signal-level scheme proposed in~\cite{mahmoodi_ICC2022_nonsym} is limited to scenarios where the global caching gain is not greater than the multiplexing gain. Hence, a general framework that can scale with the number of users without such scaling impediments is still missing in the literature.
%%%%%%%%%%%%%%%%%%%%%%%%%%%%%%%%%%%%%%%%%%%%%%%%%%%%

\subsection{Our Contribution}
\label{section:our_contribution}

A novel location-dependent multi-antenna \ac{CC} scheme is proposed in this paper, leveraging location-aware placement and delivery arrays (LAPDA) formed using a proper set of MLPDAs (c.f.,~\cite{MLPDA_ISIT2022}). In the considered system setup, users are equipped with dedicated cache memories and can roam freely within the application environment. A location-aware, non-uniform memory allocation strategy similar to~\cite{mahmoodi2022asymmetric_arxiv_TWC2023} is employed to ensure that users in areas with poor wireless link quality do not experience excessive delivery times. Due to different amounts of memory allocated to different locations, multiple location-dependent MLPDAs are utilized to define which part(s) of every file should be cached by each user. As a result, the number of cached subfiles for each location-dependent content could be different, and hence, a different transmission schedule may be needed to deliver the missing subfiles requested at each location. To handle this requirement, a file-mapping process is devised where all requested files in different locations are divided into equal numbers of file fragments (of various sizes). A common MLPDA is then used to deliver all requested file segments to all users, regardless of their location.

Parts of this paper have been published in our previous work~\cite{mahmoodi_ICC2022_nonsym}. In this paper, one important limitation of~\cite{mahmoodi_ICC2022_nonsym} is addressed, namely the requirement for a larger spatial multiplexing gain than the global caching gain. This limitation arose from the utilization of the \ac{CC} scheme proposed in~\cite{salehi2020lowcomplexity}, which is primarily designed for scenarios with a lower global caching gain relative to the multiplexing gain. To circumvent this limitation, we instead utilize a set of appropriately designed MLPDAs, which are demonstrated to be a general framework for signal-level \ac{CC} schemes~\cite{MLPDA_ISIT2022}, including shared caching-based approaches (e.g.,~\cite{parrinello2019fundamental,salehi2021MIMO,salehi2023multicast}). 
Serving a large number of concurrent users submerged into a collaborative XR experience within a bounded environment is an ideal scenario for the proposed delivery scheme. In this regard, we follow the XR connectivity framework proposed in our earlier work~\cite{mahmoodi2022asymmetric_arxiv_TWC2023} but utilize a simple linear transceiver design for data delivery with a fast iterative beamforming process through the use of LAPDAs that enable unicast transmissions. In particular, the unicast transmission allows a low-complexity beamformer design based on weighted max-min optimization, which is iteratively solved via Lagrangian duality. This fast beamformer design allows the proposed scheme to be applied to large networks, leading to significantly higher achievable coded caching gains compared to our previous method presented in~\cite{mahmoodi2022asymmetric_arxiv_TWC2023}. This translates into a notable performance advantage over conventional unicasting and multicasting schemes, as demonstrated by simulation results. 

%Finally, the proposed delivery method outperforms the conventional symmetric \ac{CC} schemes despite a \ac{DoF} loss. This is due 

%%%%%%%%%%%%%%%%%%%%%%%%%%%%%%%%%%%%%%%%%%%%%%%%%%%%
\subsection{Notation and structure}
Matrices and vectors are presented by boldface upper and lower case letters, respectively, and calligraphic letters are used to denote sets. For the set $\CA$ and vector $\Bv$, $|\CA|$ and $\|\Bv\|$ represent the cardinality of $\CA$ and norm of $\Bv$, respectively. For two sets $\CA$ and $\CB$, $\CA \backslash \CB$ includes the elements of $\CA$ that are not in $\CB$. Finally,  $\mathbb{W}$ denotes the set of non-negative integers, $[m]$ represents the set of integers from 1 to $m$, and $\oplus$ denotes addition in the corresponding finite field. 

%%%%%%%%%%%%%%%%%%%%%%%%%%%%%%%%%%%%%%%%%%%%%%%%%%%%
The rest of this paper is organized as follows.
In section~\ref{sec:sysmodel}, we describe our location-based system model. A two-phase cache placement scheme comprised of memory allocation and cache arrangement processes is described in section~\ref{sec:cache_placement}, while section~\ref{sec:delivery} discusses the delivery procedure. In section~\ref{Sec:beamforming}, weighted-max-min beamforming, tailored for the considered location-based cache placement setup, is introduced. In the end, numerical results are provided in section~\ref{sec:Simulations}, while section~\ref{sec:conclusions} concludes the paper.
%%%%%%%%%%%%%%%%%%%%%%%%%%%%%%%%%%%%%%%%%%%%%%%%%%%%

\section{System Model}
\label{sec:sysmodel}
%A system model similar to~\cite{mahmoodi2022asymmetric_arxiv_TWC2023} is considered, now with new placement and transmission schemes designed to support large networks without practical considerations. We briefly review the assumptions used in the previous study to provide a comprehensive analysis. 
A downlink scenario is considered where a server with $L$ transmit antennas serves $K$ single-antenna, cache-enabled users.\footnote{Here, $L$ refers to the attainable spatial multiplexing gain at the transmitter, which is upper-bounded by the real number of antennas. Nevertheless, `antenna count' is used throughout the text for simplicity.} The users are located within a bounded environment, such as a gaming hall, an operating room, or an exhibition hall. The system model is quite similar to our previous study in~\cite{mahmoodi2022asymmetric_arxiv_TWC2023} but with new placement and transmission schemes designed to support large networks with improved scalability. % such as the number of transmit antennas and the memory available to users. %\footnote{The system model can be extended to multi-antenna receivers following a similar approach proposed in~\cite{salehi2021lwsa}.} 
Let $\CK = [K]$ denote the set of users with limited memory capacities who can navigate within the coverage area. Users are assumed to request data from the server, depending on their location and application requirements. The environment is partitioned into $S$ single transmission units (STU), wherein a distinct 3D image is required to reconstruct the $360$-degree spherical virtual viewpoint around the user at each STU~\cite{salehi2022enhancing}. As a small example, Figure~\ref{fig:system-model} represents a simple application environment with eight STUs, where $\CS$ denotes the set of STUs. %In this paper, we use the term \emph{state} interchangeably with the term STU.

\begin{figure}
\centering 
    \includegraphics[width=1\columnwidth,keepaspectratio]{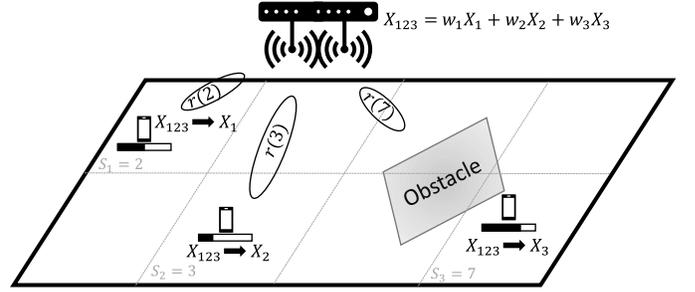}
    \caption{An application environment with $K=3$ users, split into $S=8$ STUs. $r(s)$ is the STU-specific achievable rate and $r(3) > r(2) > r(7)$. $X_{123}$ is the transmitted message, $X_i$ and ${\bf{w}}_i$ represents the data part intended for user $i$ and its corresponding precoder, respectively. The black bar below each user indicates how much of the requested data is cached.}
    \label{fig:system-model}
\end{figure}

The STU mapping is designed so that the wireless channel quality can be assumed to be almost the same for all points within a given STU. We also assume the 3D image within each STU can be decomposed into \textit{static} and \textit{dynamic} components~\cite{mahmoodi2022asymmetric_arxiv_TWC2023,salehi2022enhancing}. An example of such decomposition is shown in Figure~\ref{fig:static_dynamic_decomposition}.
A proper modeling structure, such as the one described in~\cite{flashback_2016_VR_static_dynamic_support}, would allow users to cache the entire static part and a significant portion of the dynamic part in advance~\cite{mahmoodi2022asymmetric_arxiv_TWC2023}. In this paper, we concentrate on the efficient delivery of this cacheable part of the content.\footnote{Due to the interaction of objects in the virtual world, the BS must also provide control data to aid users in reconstructing the dynamic content. However, such control data is considered to cause a fairly minor overhead and is omitted in this paper.}

\begin{figure}[t]
     \centering
      \includegraphics[width=1\columnwidth,keepaspectratio]{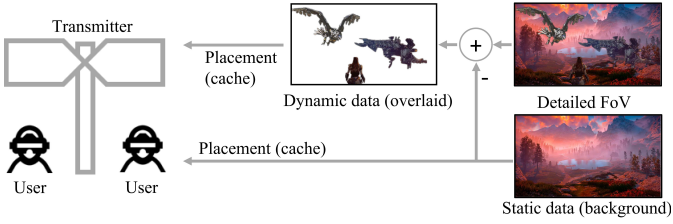}
        \caption{XR Data decomposition into static and dynamic parts.}
        \label{fig:static_dynamic_decomposition}
\end{figure}

Denote $W(s)$ as the (cacheable part of) file required for reconstructing the detailed FoV in STU $s \in \CS$ and, without loss of generality, assume $|W(s)| = F$ bits for every $s\in\CS$. Unless otherwise stated, this paper considers normalized data units, and $F$ is dropped in subsequent notations. System operation consists of two phases: \emph{cache placement} and \emph{delivery}. During the placement phase, each user $k$, equipped with a cache memory of size $M$ (normalized) bits, stores a message $Z_{k} = Z_{k}(W(s_{1}), \dots, W(s_{S}))$ in its cache, where $Z_{k}(\cdot)$ denotes a function of the files $W(s)$, $\forall s \in \CS$, with entropy not larger than $M$ bits. 

During the delivery phase, users located within the application environment request missing data from the server to reconstruct the FoV of their current locations. Specifically, a request vector $\Bd:= \{d_k \ | \ d_k \in \CS$, $\forall k \in \CK\}$ is first collected at the BS, where $W_{d_k} \equiv W(s_{d_k})$ is the file requested by user $k$ in STU $s_{d_k}$. To deliver missing parts of the files in $\Bd$, the BS then transmits several precoded messages $\Bx_{\CU}$ at different intervals, where $\CU \subseteq \CK$ denotes the set of users receiving (a part of) their requested data from $\Bx_{\CU}$. The number of precoded messages (and hence, the number of different user sets $\CU$) depends on the underlying \ac{CC} scheme. 
%For the sake of simplicity, we consider data transmission in a single interval and remove $\CU$ in our subsequent notations.
%Note that distinct sets of users $\CU \subseteq \CK$ are served at each interval, where $\Bx_{\CU}$ indicates the message intended for user-set $\CU$. In the following, since every transmission follows the same procedure, we outline the delivery process for a particular transmission that occurs within a single interval.
%
%Before going through a detailed discussion on the delivery process in Section~\ref{sec:delivery}, let us briefly introduce the signal model for the transmit and receive signal model.
Every message $\Bx_{\CU}$ comprises several unit power codewords $x_k$, where $x_k$ contains useful data for a user $k \in \CU$. Thus, $\Bx_{\CU}$ is built as
%\begin{equation}
$    \Bx_{\CU} = \sum_{k \in {\CU}} \Bv_{k}x_{k} $,
%\end{equation}
where $\Bv_{k} \in \mathbb{C}^{L}$ denotes the precoding vector dedicated to codeword $x_k$. To be specific, $\Bv_k$ is designed to suppress the interference caused by $x_k$ on a subset of users in $\CI_k \subset \CU$ that can not remove the interference by their cache contents.
%interference on a set of users $\CI_k \subset \CU$, whereas the remaining users $\CU \setminus \{\CI_k \cup k\}$ utilize their cache content to remove $x_k$ from their received signal.
%
After the transmission of $\Bx_{\CU}$, every user $k \in \CU$ receives
\begin{equation}\label{eq:recieved_signal_sysmodel}
    y_{k} = \Bh_{k}\herm\sum_{\bar{k} \in \CU}\Bv_{\bar{k}} x_{\bar{k}} + z_k,
\end{equation}
where the channel vector between the BS and user $k$ is denoted by $\Bh_{k} \in \mathbb{C}^{L}$, and $z_k \sim \mathbb{CN}(0,N_0)$ represents the additive white Gaussian noise.
Note that both the local cache content $Z_k$ and the received signals from the wireless channel over different time intervals are used at the decoder of user $k$ to reproduce the requested file $W_{d_k}$. Moreover, the instantaneous channel state information at the transmitter (CSIT) is assumed to be available during the \emph{delivery phase}, which is utilized for beamformer design and rate allocation.\footnote{CSIT measurement is feasible through reciprocal reverse link pilot measurements assuming data delivery is carried out within the channel coherence time. A detailed discussion on CSI acquisition in CC networks can be found in~\cite{salehi2020lowcomplexity}.}

Finally, as discussed in~\cite{mahmoodi2022asymmetric_arxiv_TWC2023}, an approximate throughput estimate, e.g., based on the statistics collected from previous application runs, is required for proper location-dependent cache placement. Unlike the delivery phase, it is not possible to calculate instantaneous achievable rates during the placement phase. This is because important information such as concurrently scheduled users, their locations, channel conditions, and precoding algorithms is not yet available. The aim of the approximation is to have a relative rate difference among different STUs available to allocate distinct portions of memory to them accordingly. Intuitively, to avoid extensive transmission times for users with poor connectivity, data needed at STUs with the lowest approximated rates should occupy most of the memory. To this end, we use the following rate approximation originally proposed in~\cite{mahmoodi2022asymmetric_arxiv_TWC2023} 
\begin{equation} \label{eq: state-rate}
    {r}(s) = \frac{\Omega}{F}C_p\mathbb{E}\big[\log(1+\frac{P_T \ \|\Bh_{k_s}\|^2}{N_0})\big] \quad \text{[files/second]},
\end{equation}
where $C_p$ is a pre-log scaling factor containing any practical overhead, $P_T$ is the transmission power, $\Omega$ is the communication bandwidth, and $\Bh_{k_s} \in \mathbb{C}^{L}$ is the channel vector between the server and a user $k$ located in STU $s$. Note that the expectation is taken over all user locations and channel realizations in STU $s$ (c.f.~\cite{mahmoodi2022asymmetric_arxiv_TWC2023} for more details).

\section{Cache Placement}
\label{sec:cache_placement}
A PDA-based location-dependent cache placement scheme comprising two consecutive processes, \emph{memory allocation} and \emph{cache arrangement}, is used in this paper. The memory allocation process is similar to~\cite{mahmoodi2022asymmetric_arxiv_TWC2023} and prioritizes content requested in locations with poor wireless connectivity in order to mitigate excessive delays during the content delivery phase. Given the result of the memory allocation process, the cache arrangement process is then used to clarify what data parts should be cached by each user. One of the key novelties of this paper is the introduction of a PDA-based cache arrangement process that allows the overall scheme to be scalable w.r.t various network parameters. Nevertheless, for clarity, the memory allocation process of~\cite{mahmoodi2022asymmetric_arxiv_TWC2023} is also briefly explained in the following.

%This paper proposes a location-dependent cache placement scheme represented by a well-defined PDA. However, unlike other existing PDA-based schemes, our cache placement phase comprises two consecutive processes, \emph{memory allocation} and \emph{cache arrangement}. The key novelty of our scheme lies specifically in the cache arrangement process. To elaborate further, our proposed placement scheme initially applies a memory allocation process, which was first introduced in previous work~\cite{mahmoodi2022asymmetric_arxiv_TWC2023}. This process is intended to prioritize content relevant to locations with poor wireless connectivity. The goal is to mitigate excessive delays during content delivery in such areas. A PDA-based cache arrangement process is introduced in the following steps to handle uneven memory allocations to different files that result from the earlier memory allocation phase. In the following, the memory allocation process originally introduced in~\cite{mahmoodi2022asymmetric_arxiv_TWC2023} is briefly outlined to ensure clarity.
%
%Notably, the novelty lies in the cache arrangement process. The proposed placement scheme first employs an existing memory allocation process first introduced in~\cite{mahmoodi2022asymmetric_arxiv_TWC2023}, wherein the goal is to prioritize content relevant to locations with poor wireless connectivity to prevent excessive delays during the delivery.
%
%This differs from the time minimization process proposed in Section~\ref{sec:delivery}, which relies on an optimal beamformer design. 

\noindent\textbf{Memory allocation~\cite{mahmoodi2022asymmetric_arxiv_TWC2023}:}
Real-time applications (e.g., XR) typically require a bounded delivery time. Excessive delivery delays can be circumvented by reserving a larger share of memory to store data requested in poor connectivity areas. Accordingly, the memory allocation process specifies the normalized cache size $m(s)$ reserved for storing (a fraction of) each STU-specific file $W(s)$ at every user. 
This paper assumes no prior knowledge about the users' spatial locations during the subsequent delivery phase; hence, we consider uniform access probability for all STUs during the placement phase.\footnote{The placement efficiency can be further improved by using prior knowledge about the access likelihood for each STU.} 
%Let us use $m(s)$ to denote the normalized cache size at each user allocated to store (a fraction of) $W(s)$. 
%Since a user at STU $s$ has $m(s)$ fraction of the file $W(s)$ in its cache, to reconstruct the FoV of STU $s$, the remaining $1-m(s)$ fraction must be delivered to the user over the wireless link.
%
%In Sec~\ref{sec:delivery}, we show that if $m(s)$ values are known, the total delivery time of the proposed scheme, denoted by $T_T$, can be approximated as

%We use a reverse method to approximate optimal $m(s)$ values.
%In Section~\ref{sec:delivery} (Lemma~\ref{lemm:aprox_delivery_time}), we prove 
Following~\cite{mahmoodi2022asymmetric_arxiv_TWC2023}, if $m(s)$ values are known, the total delivery time $T_T$ can be approximated as
\begin{equation} \label{eq: aprrox_deliv_time}
    \hat{T}_T = \frac{K}{K\bar{m}+L}\max_{s \in \CS}\frac{1-m(s)}{r(s)} \quad \textrm{[seconds]},
\end{equation}
where $\bar{m} = \min\limits_{s \in \CS}m(s)$ is the least allocated memory for a STU
%given the non-uniform memory allocation, 
and $r(s)$ is the approximated rate at STU $s$ (c.f. Eq.~\eqref{eq: state-rate}). 
%The validity of the approximation in equation~\eqref{eq: aprrox_deliv_time} has been formally established in Section~\ref{sec:delivery}, Lemma~\ref{lemm:aprox_delivery_time}, where a comprehensive analysis of the entire delivery process is conducted. Thus, for brevity, we will skip the details in this section. 
%As a quick explanation to~\eqref{eq: aprrox_deliv_time}, t
The $K\bar{m} +L$ term in the denominator of~\eqref{eq: aprrox_deliv_time} approximates the achievable DoF for the non-uniform memory allocation scenario (note that for the uniform allocation, the DoF is upper bounded by $K\frac{M}{S}+L$~\cite{MLPDA_ISIT2022}), and the term $K \max\limits_{s \in \CS}\frac{1-m(s)}{r(s)}$ approximates the worst-case delivery time across all the STUs when $K$ users are served simultaneously. 
In order to find approximate $m(s)$ values that minimize the expected delivery time, we first rewrite~\eqref{eq: aprrox_deliv_time} as $\hat{T}_T = \frac{1}{\bar{m}+\frac{L}{K}}\max_{s \in \CS}\frac{1-m(s)}{r(s)}$, %then, the fixed term $\frac{L}{K}$ is substituted by a general term $\phi$, which will be efficiently optimized (c.f., Algorithm~\ref{Alg:placement}). 
and then, 
%to minimize the delivery time for any possible realizations of user locations, 
formulate the memory allocation process as the following linear fractional programming (LFP) problem:
\begin{equation} 
\begin{aligned}
\label{cache-allocation}
&\min_{m(s), \; \gamma \ge 0, \; {m} \ge 0} \quad  \frac{\gamma}{{m}+\frac{L}{K}}\\
&\textrm{s.t.} \quad  \frac{1-m(s)}{r(s)} \leq \gamma, \ \forall s \in \CS, \\ &{m} \leq m(s), \ \forall s \in \CS,  \quad \sum_{s \in \CS} m(s) \leq M.
\end{aligned}
\end{equation}
Note that at the optimal solution to~\eqref{cache-allocation}, $ m = \bar{m} = \min\limits_{s \in [S]}m(s)$. Using the Charnes-Cooper transformation, the LFP in~\eqref{cache-allocation} can be reformulated as a linear programming problem and solved efficiently~\cite{mahmoodi2022asymmetric_arxiv_TWC2023}. For ease of exposition, we assume that $Km(s)$ is a positive integer for all $s\in\CS$ throughout the text. %In Appendix~\ref{sec:Appendix_non_int}, we explain how to address 
The non-integer $Km(s)$ case is addressed in Appendix~\ref{sec:Appendix_non_int} using time-sharing, which is an alternative method that surpasses the performance of the approach proposed in~\cite{mahmoodi2022asymmetric_arxiv_TWC2023}. 
%The case where these assumptions are violated is addressed in Appendix~\ref{sec:Appendix_non_int}.

% Using Charnes-Cooper transformation~\cite{charnes1962programming}, this problem can be reformulated as an equivalent linear program (LP) as follows: 
% \begin{equation} 
% \begin{aligned}
% \label{cache-allocation_dif_conv}
% &\min_{m^{'}(s), \; \gamma^{'} \ge 0, \; {m}^{'} \ge 0 \; \xi \ge 0} \quad  \gamma^{'}\\
% &\textrm{s.t.} \quad  \frac{\xi-m^{'}(s)}{r(s)} \leq \gamma^{'}, \ \forall s \in \CS, \\ &\bar{m}^{'} \leq m^{'}(s), \ \forall s \in \CS, \quad \sum_{s \in \CS} m^{'}(s) \leq M\xi,
% \quad {m}^{'} + \frac{L}{K}\xi = 1.
% \end{aligned}
% \end{equation}
% After solving~\eqref{cache-allocation_dif_conv}, the actual allocated memory would be $m(s) = m^{'}(s)/\xi, \forall s$, which also implies ${m} = {m}^{'}/\xi$. 

%Moreover, integer $Km(s)$ is also assumed when driving~\eqref{eq: aprrox_deliv_time}, resulting in a performance loss due to non-integer values in the actual implementation. Thus, optimizing $\frac{L}{K}$ using Algorithm~\ref{Alg:placement} minimizes the performance loss.

\noindent\textbf{Cache arrangement:}
We utilize a location-aware placement delivery array (LAPDA) to store data fragments of files in users' cache memories. Let $\boldsymbol{\CQ}$ denote a specific LAPDA consisting of a set of $S$ STU-specific MLPDA matrices $\BQ_s$, $s \in \CS$, that are interrelated with an extra cross-matrix condition that ensures data delivery is possible with the given non-uniform memory allocation.
%on all the matrices' entries. 
%This condition ensures that a larger group of users will cache files with larger allocated memory. It specifically mandates that users who cache files with lower allocated memory must be a subset of the users caching files with larger allocated memory. 
Before going through a detailed explanation of LAPDA, let us first review the general definition of MLPDAs.%, followed by our definition of LAPDAs.

\begin{definition}\label{MLPDA_definition_new}
A $(L,K,F_s,Z_s,N_s)$ MLPDA $\BQ_s = \left[\BQ_s(f_s, k)\right], \ f_s \in [F_s], \ k \in [K],$ is a $F_s \times K$ matrix whose elements include the specific symbol “$\ast$” and $N_s$ positive integers $\{1, 2, \dots, N_s\}$. For positive integers $L, K, F_s$, and $Z_s$, $\BQ_s$ satisfies~\cite{MLPDA_ISIT2022}:
    \begin{enumerate}
    \item[C1.] The symbol “$\ast$” appears $Z_s$ times in each column, such that $\frac{Z_s}{F_s} = m(s)$;
    \item[C2.] Each integer $n \in [N_s]$ appears at least once in the matrix;
    \item[C3.] Each integer appears at most once in each column; 
    \item[C4.] For every $n \in N_s$, if we define $\CU(n,s):=\{k \ | \ \exists f_s \in [F_s], \ \BQ_s(f_s,k) = n\}$, $\CF(n,s):= \{f_s \ | \ \exists k \in [K], \ \BQ_s(f_s,k) = n\}$, 
    %as the set of users, and the set of subfiles served in $n$'th transmission, respectively. Define 
    and $\BQ_s^n$ to be a sub-matrix of $\BQ_s$ comprised of all rows and columns containing $n$ (i.e., $\BQ_s^n = [\BQ_s(f_s,k)], f_s \in \CF(n,s)$, $ k \in \CU(n,s)$), the number of integer entries in each row of $\BQ_s^n$ is less than or equal to $L$, i.e.,
        \begin{equation*} \scriptsize
            |\{k \ | \ k \in \CU(n,s), \BQ^n_s(f_s,k) \in [N_s]\}| \leq L, \forall f_s \in \CF(n,s).
        \end{equation*}
    \end{enumerate}
\end{definition}

As discussed in~\cite{MLPDA_ISIT2022}, the $(L,K,F_s,Z_s,N_s)$ MLPDA $\BQ_s$ uniquely identifies a placement-delivery strategy for a MISO network with $K$ cache-enabled users, a coded caching gain of $t_s =Km(s)$, and a spatial multiplexing gain of $L$. In this regard, each file is first divided into $F_s$ subpackets, from which all subpackets $f_s \in [F_s]$ are stored by all users $k \in [K]$ if $\BQ_s(f_s,k) = \ast$. The delivery phase then consists of $N_s$ transmissions, where at transmission $n \in [N_s]$, subpackets $f_s \in \CF(n,s)$ are sent to users $k \in \CU(n,s)$ if $\BQ_s(f_s,k) = n$. Note that condition C1 ensures that the memory constraints are met, C2 prevents empty transmission, C3 removes the need for successive interference cancellation (SIC), and C4 ensures that any interference that cannot be removed with cache content is suppressed by a proper precoder.

\begin{algorithm}[t]
\footnotesize
	\caption{Location-based cache placement}
	\begin{algorithmic}[1]
 \Procedure{Memory Allocation}{}%
 %\State $T_b \gets \inf$
    %\ForAll{$\phi \in \{0, 0.1\frac{L}{K}, \dots, \phi_{\text{max}}\}$}
    %\State $\{\hat{m}(s) \} =$ The result of the LP problem in~\eqref{cache-allocation}
    %\If{$\frac{K}{\floor{ K\min\limits_{s \in [S]}\hat{m}(s)} + L}\max_{s \in \CS}\frac{1-\hat{m}(s)}{r(s)} \leq T_b$}
    %\State $T_b \gets \frac{K}{\floor{ K\min\limits_{s \in [S]}\hat{m}(s)} + L}\max_{s \in \CS}\frac{1-\hat{m}(s)}{r(s)}$
    \State Find $\{m(s) \}$ by solving Eq.~\eqref{cache-allocation}%\gets \{\hat{m}(s) \}$
    %\EndIf
    %\EndFor
 \EndProcedure 
		\Procedure{CACHE\_ARRANGEMENT}{}
		
		\ForAll{$s \in \CS$} 
		
		%\State Build LAPDA $\BQ = \{\BQ_s\}$ based on Definitions~\ref{MLPDA_definition_new} and~\ref{LAPDA_definition}.
		
		\State $W(s) \rightarrow \{W_{f}(s) \; | \; \forall f \in [F_s]\}$
		
		\ForAll{$f \in [F_s]$}
		\ForAll{$k \in \CK$}
		\If{$\BQ_s(f,k) = \ast$}
		    \State Put $W_{p}(s)$ in the cache of user $k$
		\EndIf
		\EndFor
		\EndFor
		\EndFor
		\EndProcedure 
	\end{algorithmic}
	\label{Alg:placement}
\end{algorithm}

For the proposed location-dependent cache placement where each STU $s \in \CS$ has a possibly different allocated memory portion $m(s)$, it is necessary to use a different MLPDA $\BQ_s$ for each state to satisfy the memory constraint. This is different from conventional MLPDA schemes that use a single MLPDA to store all library files. In fact, given a set of MLPDAs $\{\BQ_s\}$, the files for every STU $s \in \CS$ are first divided into $F_s$ subfiles, where $F_s$ can vary for each STU. Then, for every STU $s \in \CS$, every user $k \in \CK$ caches all subfiles $f_s \in [F_s]$ if $\BQ_s(f_s,k) = \ast$. Algorithm~\ref{Alg:placement} summarises the placement process for a set of MLPDA $\{\BQ_s\}$.

Using a distinct $\BQ_s$ for each STU results in an unequal number of subfiles $F_s$, number of cached data elements $Z_s$, and time slots $N_s$ to deliver location-dependent missing data. Hence, an additional cross-matrix condition is added to the conventional MLPDA definition to ensure a feasible delivery scheme for the proposed non-uniform placement. Accordingly, a proper LAPDA $\boldsymbol{\CQ}$ is defined as follows.

%
%Now, for our location-dependent content delivery scenario where each STU $s$ has its unique allocated memory portion $m(s)$, we employ a different MLPDA $\BQ_s$ for placement purposes as outlined in Algorithm~\ref{Alg:placement}. 
%
%add here a brief intro how... the location dependency is affecting the placement...

\begin{definition}\label{LAPDA_definition}
    A $(L,K,\{F_s\},\{Z_s\},\{N_s\})$ LAPDA $\boldsymbol{\CQ}$ is a set of $S$ number of $(L,K,F_s,Z_s,N_s)$ MLPDAs  $\BQ_s$ where for every $(s,\hat{s}) \in \CS$ for which $m(s) > m(\hat{s})$ we have
    \begin{equation*}
    \begin{aligned}
        &\forall f_{\hat{s}} \in [F_{\hat{s}}], \  \exists f_s \in [F_s] : \CB_{f_{\hat{s}}} \subseteq \CB_{f_s}, \\ &\forall f_s \in [F_s], \  \exists f_{\hat{s}} \in [F_{\hat{s}}] : \CB_{f_{\hat{s}}} \subseteq \CB_{f_s},
    \end{aligned}
    \end{equation*}
    where $\CB_{f_s} = \{k \ | \ k \in [K], \ \BQ_s(f_s,k) = \ast\}$ is the set of columns including $\ast$ in their $f_s$'th row.
\end{definition}

The following example illustrates the entire cache placement process, including memory sharing and cache arrangement.
In Section~\ref{sec:delivery}, we propose a novel delivery algorithm tailored for the above described non-uniform cache placement that achieves a significant coded caching gain, similar to the location-dependent scheme of~\cite{mahmoodi2022asymmetric_arxiv_TWC2023}, but now applicable to much larger networks. Moreover, in Appendix~\ref{sec: Appendix A}, we discuss how the extra condition in Definition~\ref{LAPDA_definition} ensures the feasibility of the delivery scheme.
%
%Section~\ref{sec:delivery} of this paper describes a novel delivery algorithm that achieves a proper coded caching gain, similar to the location-dependent scheme of~\cite{mahmoodi2022asymmetric_arxiv_TWC2023}, 
%based on users' locations, %despite 
%given the non-uniform cache placement described above. In Appendix~\ref{sec: Appendix A}, we discuss how the extra condition in Definition~\ref{LAPDA_definition} ensures the feasibility of the delivery scheme.
%As shown in Appendix~\ref{sec: Appendix A}, the condition in Definition~\ref{LAPDA_definition} is imposed to ensure the feasibility of the proposed delivery scheme. Lastly, the 

%
%Note that numerous works can be represented in MLPDA form (e.g.,~\cite{MLPDA_ISIT2022,lampiris2018adding,salehi2020lowcomplexity}) and can be readily translated to LAPDA defined in Definition~\ref{MLPDA_definition}. The location-dependent cache placement is done using $\BQ_s$ matrices. In this regard, the files for every state $s$ are first divided into $F_s$ subfiles. Then, for every state $s \in S$, user $k \in [K]$ caches subfile $f \in [F_s]$ if $\BQ_s(f,k) = \ast$. It is worth mentioning that unlike other existing works~\cite{MLPDA_ISIT2022}, where cache placement is the same for all files, here, cache placement for each file $W(s)$ is different from the others. In other words, files for different states are handled differently. The following example illustrates the cache placement process.

\begin{exmp}
\label{exmp:placement}
%Consider an application scenario 
To illustrate the proposed location-dependent cache placement, we consider an example scenario from~\cite{mahmoodi2022asymmetric_arxiv_TWC2023} with $K=4$ users and $L = 2$ transmit antennas. The environment is split into $S=5$ STUs, and for each STU, the required data size is $F=400$ Megabytes. Each user has a cache size of $900$ Megabytes; hence, the normalized cache size is $M = 2.25$ data units. The approximated normalized throughput value for each STU is as given in Table~\ref{Table:rate&cache_distribution}, where the memory allocation resulting from solving~\eqref{cache-allocation} is also shown.

Consider the following MLPDA matrices $\BQ_1$-$\BQ_5$, which satisfy the conditions in Definition~\ref{MLPDA_definition_new} for the resulting $m(s)$ values in Table~\ref{Table:rate&cache_distribution}:
\begin{equation} \label{eq:LAPDA_example} 
  \hspace{-5pt}\BQ_1 = \begin{psmallmatrix}
\ast & 1 & 1 & 2\\
\ast & 3 & 2 & 3\\
\ast & 6 & 8 & 10\\
1 & \ast & 4 & 4\\
5 & \ast & 6 & 5\\
6 & \ast & 9 & 11\\
2 & 4 & \ast & 7\\
7 & 8 & \ast & 9\\
8 & 9 & \ast & 12\\
3 & 5 & 7 & \ast\\
10 &10 &11 & \ast\\
12 &11 & 12 & \ast
\end{psmallmatrix} \hspace{-1mm}
% ,
% \BQ_2 = \BQ_4 = \begin{psmallmatrix}
% \ast & \ast & 1 & 1\\
% \ast & \ast & 2 & 2\\
% \ast & \ast & 4 & 5\\
% \ast & \ast & 5 & 8\\
% 2 & \ast & \ast & 3\\
% 3 & \ast & \ast & 4\\
% 4 & \ast & \ast & 6\\
% 7 & \ast & \ast & 7\\
% 1 & 1 & \ast & \ast\\
% 5 & 5 & \ast & \ast\\
% 6 & 6 & \ast & \ast\\
% 8 & 4 & \ast & \ast\\
% \ast & 2 & 3 & \ast\\
% \ast & 3 & 6 & \ast\\
% \ast & 7 & 7 & \ast\\
% \ast & 8 & 8 & \ast
% \end{psmallmatrix}\hspace{-1mm}
, \BQ_2 = \begin{psmallmatrix}
\ast & \ast & 2 & 1\\
1& \ast & \ast & 2\\
2 & 1 & \ast & \ast\\
    \ast &2 & 1& \ast
\end{psmallmatrix}
, \BQ_3 = \begin{psmallmatrix}
\ast & \ast & \ast & 1\\
1 & \ast & \ast & \ast\\
\ast & 1 & \ast & \ast\\
\ast & \ast & 1 & \ast
\end{psmallmatrix},
\end{equation}
$\BQ_4=\BQ_2$, and $\BQ_5=\BQ_1$. It can be seen that $\BQ_1$-$\BQ_5$ form a LAPDA according to Definition~\ref{LAPDA_definition}, and hence, can be utilized for location-dependent data delivery. In this regard, using $\BQ_1$-$\BQ_5$ for the data placement, files $W(1)$ and $W(5)$ are first divided into $12$ subfiles, from which $3$ are cached in each user's memory. Similarly, files  $W(2)$ and $W(4)$ are divided into $4$ subfiles, and $2$ of them are cached in the memory of each user. Finally, the file $W(3)$ is divided into $4$ subfiles, from which $3$ are cached in each user's memory. 
% The resulting cache placement is visualized in Figure~\ref{fig:cache pool}.
% \input{figs/PlacementFigure}
\begin{table}[t]
\centering
\begin{tabular}{|c||c|c|c|c|c|}
\cline{2-6}
 \multicolumn{1}{c|}{} & $s\!=\!1$   & $s\!=\!2$   & $s\!=\!3$   & $s\!=\!4$   & $s\!=\!5$    \\ \hline
$r(s)$                               & $3\times 10^3$    & $2\times 10^3$   & $1\times 10^3$   & $2\times 10^3$   & $3\times 10^3$    \\ \hline
$m(s)$                                & 0.25 & 0.5 & 0.75 & 0.5 & 0.25 \\ \hline
\end{tabular}
\vspace{-8pt}
\caption{Location-specific rate and memory allocation for Example~\ref{exmp:placement}.}
\label{Table:rate&cache_distribution}
\end{table}
\end{exmp}

\section{Content Delivery}
\label{sec:delivery}
During the delivery phase, users move within the application environment (hence, change their locations) over time.
In each time instance, users reveal their location-dependent file requests to the server.\footnote{Using dynamic CC techniques~\cite{salehi2021lowsubpacketization,abolpour2022coded,abolpour2023cache}, this system model can be easily modified to the case only a subset of users reveal their requests in each instant.} Without loss of generality, we consider a specific time slot, where every user $k$ in STU $s_k$ requests the file $W_{d_k} \equiv W(s_k)$ from the server to reconstruct its STU-specific FoV. Accordingly, the server builds and transmits several precoded messages to the requesting users. To reconstruct $W(s_k)$, user $k$ requires one normalized data unit, from which a portion of size $m_k \equiv m(s_k)$ units is available in its cache and the remaining part should be delivered by the server. Note that the conventional PDA-based delivery schemes are suited for scenarios where all users cache the same amount of data (e.g.,~\cite{MLPDA_ISIT2022,lampiris2018adding,salehi2020lowcomplexity}). So, they do \textit{not} apply to our case where users have cached different amounts of their requested files.

\subsection{PDA-based Delivery}
To tackle the challenge posed by uneven memory allocation, we first make a temporary assumption that all users have cached the same portion of $\hat{m} = \min_{k \in \mathcal{K}} m_k$ of their requested files, and use any conventional PDA-based delivery scheme in the literature (e.g.,~\cite{MLPDA_ISIT2022,lampiris2018adding,salehi2020lowcomplexity}) to generate a set of \emph{preliminary transmission vectors} (PTVs). These vectors are subsequently adjusted to accommodate the different file-indexing procedures employed for STU-dependent cache placement during the placement phase. Using the `$\min$' operation can limit the performance in certain scenarios when a subset of users have relatively smaller $m(s_k)$ values compared to the rest, e.g., as they are close to the transmitter. To address such cases, the concept of \emph{phantom users} has been proposed in~\cite{mahmoodi2022asymmetric_arxiv_TWC2023} to separately serve users with small $m(s_k)$ value via unicasting. However, in larger networks considered in this paper, such scenarios are less probable. This is due to the variable $m$ in the LFP formulation~\eqref{cache-allocation}, which inhibits assigning small values to any $m(s)$, especially when the ratio $L/K$ is small (e.g., for $K \gg L$ case considered in this paper).
%
%However,  in the larger networks analyzed in this paper, such scenarios are less likely to occur. This is due to the variable $m$ in the LFP formulation~\eqref{cache-allocation}, which prevents assigning samll $m(s)$ values, particularly when $L/K$ value is small. In our case, we consider networks where $K$ is much larger than $L$, resulting in a much smaller $L/K$ ratio than the scenario in~\cite{mahmoodi2022asymmetric_arxiv_TWC2023}. %Therefore, large gaps between different $m(s)$ values are less likely in our case.

For clarity, we designate $\Check{\Bx}_{\CU}$ to represent a PTV, which will be later modified to form the transmission vector $\Bx_{\CU}$. In order to build $\Check{\Bx}_{\CU}$, we use the MLPDA $\hat{\BQ}$ corresponding to the location of the user with the least available memory, i.e., $\hat{\BQ} \equiv \BQ_{s_{k^{*}}}$, $k^{*} := \arg\min_{k \in \CK} m(s_k)$. This means we need
$\hat{N} = N_{s_{k^{*}}}$ consecutive transmissions, and the PTV at time instant $n \in [\hat{N}]$ is given as
\begin{equation}\label{eq: multiplexed signal}
    \Check{\Bx}(n) \equiv \Check{\Bx}_{{\CU}(n)} = \sum\nolimits_{k \in \CU(n)}\Bv_{k}(n)W_{\Check{f}_k^n}(s_k) \; , 
\end{equation}
where 
\begin{equation} \label{eq:Target_user_set}
{\CU}(n):=\{k \ | \ \exists \Check{f}_k^n \in [F_{s_{k^*}}], \ \hat{\BQ}(\Check{f}_k^n,k) = n\}
\end{equation}
is the set of users served in the $n$'th transmission, $\Check{f}_k^n$ is the 
%delivery-index 
\emph{temporary} index of the subfile 
%transmitted 
destined to user $k \in \CU(n)$, and $\Bv_{k}(n)$ is the optimized beamforming vector dedicated to user $k$. Specifically, $\Bv_{k}(n)$ is designed to suppress $W_{\Check{f}_k^n}(s_k)$ at every user in the interference indicator set
\begin{equation}\label{eq:Interference_inc_set}
    \CI_k(n) = \{j \, | \, j \in \CU(n)\setminus k, \, \hat{\BQ}(\Check{f}_k^n,j) \neq \ast\}.
\end{equation} 
 In fact, $\CI_k(n)$ includes all users $j \in \CU(n)\setminus k$ that do not have $W_{\Check{f}_k^n}(s_k)$ available in their cache. 

Since $\Check{\Bx}(n)$ is built using matrix $\hat{\BQ}$ but data placement is done using the set of matrices $\BQ_s$, the temporary index $\Check{f}_k^n$ may \emph{not} coincide with the missing subfiles of the files requested by every user $k \in \CU(n)$. Example~\ref{exmp:virt_network} clarifies this statement.

\begin{exmp}
\label{exmp:virt_network}
Consider the network in Example~\ref{exmp:placement}, for which the cache placement is given in~\eqref{eq:LAPDA_example}. Consider a specific time instant $n$ with the following user-to-STU associations: $s_1 = 1$, $s_2 = 2$, $s_3 = 3$, and $s_4 = 4$. Denoting the set of requested sub-files for user $k$ with $\CM_k$ and assuming $A \equiv W(1)$, $B \equiv W(2)$, $C \equiv W(4)$, $D \equiv W(5)$, we have
\begin{equation}\label{eq:exmp_req_subfiles}
    \begin{aligned}
    \CM_1 &= \{ A_{4}, A_{5}, A_{6}, A_{7}, A_{8}, A_{9}, A_{10}, A_{11}, A_{12} \}, \\ 
    \CM_2 &= \{ B_{3}, B_{4}\}, \quad
    \CM_3 = \{ C_{4}\}, \quad
    \CM_4 = \{ D_{1}, D_{2}\}. 
    \end{aligned}
\end{equation}
Note that the subfiles of $A, B, C, D$ are $1/12,1/4,1/4,1/4$ data units in size, respectively. Here, the minimum available amount of the requested data in cache belongs to user $1$; hence, $\hat{\BQ} \equiv \BQ_1$, and we serve all users in $\hat{N} = N_{s_1} = 12$ time slots. For example, the first PTV is built as
\begin{equation*}
    \begin{aligned}
        \Check{\Bx}(1) &= \Bv_{1}(1) A_{\Check{4}}+\Bv_{2}(1) B_{\Check{1}}+\Bv_3(1) C_{\Check{1}}, %\ \quad \Check{\Bx}(2) = \Bv_1(2) A_{\Check{7}}+\Bv_3(2) C_{\Check{2}}+\BV_4(2) D_{\Check{1}},%\\
        %\Check{\Bx}(3) &= \Bv_1(3) A_{\Check{10}}+\Bv_2(3) B_{\Check{2}}+\Bv_4(3) D_{\Check{2}}, \quad \Check{\Bx}(4) = \Bv_2(4) B_{\Check{7}} +\Bv_3(4) C_{\Check{4}}+\Bv_4(4) D_{\Check{4}}, 
    \end{aligned}
\end{equation*}
and the rest can be built accordingly. Now, considering $\Check{\Bx}(1)$, temporary file indices for users~1,~2, and~3 are~4,~1, and~1, respectively. However, from~\eqref{eq:exmp_req_subfiles}, users~2 and~3 already have $B_1$ and $C_1$ in their cache memories. Hence, we must carry out an appropriate index mapping process in PTVs before transmission. As a side note, using~\eqref{eq:Interference_inc_set}, one can easily verify that the interference indicator sets for PTV $\Check{\Bx}(1)$ are $\CI_{1}(1) = {3}$, $\CI_{2}(1) = {3}$, and $\CI_{3}(1) = {2}$.
\end{exmp}

To adjust temporary file indices in PTVs, we first note that every user $k$ appears $\hat{F} - \hat{Z}$ times in all PTVs, where $\hat{Z} \equiv Z_{s_{k^*}}$ and $\hat{F} \equiv F_{s_{k^*}}$. However, in practice, every user $k$ needs $F_{s_k}-Z_{s_k}$ subfiles to construct its FoV. Hence, to uniformly map the missing subfile indices $f_{s_k} \in [F_{s_k}]$ into temporary PTV indices $\Check{f}_k^n \in [\hat{F}]$, we need to divide each subfile into $D_k = \alpha\frac{\hat{F}-\hat{Z}}{F_{s_k}-Z_{s_k}}$ \emph{file-fragments}, where $\alpha$ is a normalizing coefficient guaranteeing $D_k$ is an integer for every user~$k$ (for example, we may set $\alpha$ to be the smallest common multiplier of all $F_{s_k}-Z_{s_k}$ values). We use $W^{\phi}_{f_{s}}(s_k)$, $\phi \in [D_k]$, to represent the file-fragments resulting from subfile $ W_{f_{s}}(s_k)$.

After the division of subfiles into file-fragments, the transmission vectors $\Bx(n)$ are obtained from PTVs $\Check{\Bx}(n)$ by replacing each subfile $ W_{\Check{f}_k^n}(s_k)$ with $\underset{{m \in [\alpha]}}{\prod}\left(W^{q}_{\BP_k(n,m)}(s_k)\right)$, where $\prod$ denotes bit-wise concatenation, $\BP_k$ is the $\hat{N} \times \alpha$ \emph{user-specific file-fragment index matrix} (explained shortly), and $q(W_{\BP_k(n,m)}(s_k))$ represents the file-specific counter corresponding to the $q$'th fragment of the file $W_{\BP_k(n,m)}(s_k)$. After initializing with $q=1$, it is incremented by one each time a fragment of $W_{\BP_k(n,m)}(s_k)$ is assigned to the transmission vector. Finally, $\Bx(n)$ is expressed as
\begin{equation}\label{eq:original data vector}
    \Bx(n) = \sum_{k \in \CU(n)}\Bv_{k}(n)\underset{{m \in [\alpha]}}{\prod}\left(W^q_{\BP_k(n,m)}(s_k)\right) . \; 
\end{equation}
\begin{algorithm}[t]
\footnotesize
    \caption{File-Fragment Index Matrix}
    \label{alg:File_ind_gen}
    \begin{algorithmic}[1]
        \Function{\textsc{Index-generator}}{$\{\BG_k\}, \hat{\BQ}$}
        \ForAll{$k \in [K]$}
            \State $\BP_k \gets \mathbf{0}$
        \EndFor
         \ForAll{$n \in [\hat{N}]$}
                    \ForAll{$k \in \CU(n)$}
                        \State $\chi \gets \Check{f}_k^n$
                        \ForAll{$m \in [\alpha]$}
                            \ForAll{$\omega \in [F_{s_k}]$}
                                \If{$g_{\chi,\omega,k} > 0$}
                                    \State $\BP_{k}(n,m) \gets \omega$
                                    \State $g_{\chi,\omega,k} \gets g_{\chi,\omega,k}-1$
                                    \State Break;
                                \EndIf
                            \EndFor
                        \EndFor
                    \EndFor
          \EndFor
            \State \Return $\{\BP_k\}$
        \EndFunction
    \end{algorithmic}
\end{algorithm}
Matrix elements $\BP_k(n,m) \in [F_{s_k}]$ are designed such that 1) all the missing subfiles are delivered to all the users and 2) the cache-aided interference cancellation is performed correctly. To build $\BP_k$, we first form $K$ user-specific File-Mapping (FM) matrices $\BG_k, \forall k \in [K]$, with size $\hat{F} \times F_{s_k}$, to map 
missing subfile indices to temporary PTV indices. Denoting the $i$'th row and $j$'th column of the matrix $\BG_k$ by $g_{i,j,k}$, it represents the number of file-fragments of the subfile $W_{j}(s_k)$ that should be included in the concatenation process while building $\Bx(n)$ in~\eqref{eq:original data vector}, if the corresponding PTV $\Check{\Bx}(n)$ includes $W_{i}(s_k)$ (recall that $j \in [F_{s_k}]$ and $i \in [\hat{F}]$). Accordingly, the FM matrices $\BG_k, k \in [K]$ are defined as follows:%constructed to satisfy the following set of conditions:

\begin{definition} \label{FM_matrix_definition}
A user-specific FM matrix $\BG_k = \left[g_{i,j,k}\right], \ \forall i \in [\hat{F}], j \in [F_{s_k}], k \in [K]$ is a $\hat{F} \times F_{s_k}$ matrix whose elements are comprised of non-negative integers, and satisfy the followings:
\begin{enumerate}
  \item[C1.] 
  Denoting the temporary index set (TIS) of user $k$ as $\Check{\CN}_{k}$ (i.e., $|\Check{\CN}_{k}| = \hat{F}-\hat{Z}$), every row of $\BG_k$ not included in $\Check{\CN}_{k}$ should be zero, i.e.,
  \begin{equation*}
      g_{i,j,k} = 0, \quad \forall k \in [K],  \ \forall i \in [\hat{F}] \setminus \Check{\CN}_{k},  \ \forall j \in [F_{s_k}].
  \end{equation*}
  \item[C2.] 
  For every user $k$, denoting the set of its missing subfile indices in STU $s_k$ as a requested index set (RIS) $\CN_k$ (i.e., $|\CN_k| = F_{s_k}-Z_{s_k}$), every column of $\BG_k$ not included in $\CN_k$ should be zero, i.e., 
  \begin{equation*}
      g_{i,j,k} = 0, \quad \forall k \in [K], \ \forall j \in [F_{s_k}] \setminus \CN_k, \ \forall i \in [\hat{F}].
  \end{equation*}
  \item[C3.] Based on Definition~\ref{MLPDA_definition_new}, each temporary index $i \in \Check{\CN}_{k}$ appears only once in PTVs. However, every user $k$ needs a total number of $D_k (F_{s_k}-Z_{s_k}) = \alpha(\hat{F}-\hat{Z}) = \alpha|\Check{\CN}_{k}|$ file-fragments. Thus, $\alpha$ file-fragments must be considered for the concatenation process in~\eqref{eq:original data vector}, i.e., 
  \begin{equation*}
      \sum_{j \in \CN_k}g_{i,j,k} = \alpha, \quad \forall k \in [K], \ \forall i \in \Check{\CN}_{k}.
  \end{equation*}
  \item[C4.] As discussed above, each missing subfile of user $k$ is divided into $D_k$ file-fragments. To make sure that all these file-fragments are delivered, the sum of the elements in every column of $\BG_k$ that is included in RIS must be equal to $D_k$, i.e., 
  \begin{flalign*}
      \sum_{i \in \Check{\CN}_k}g_{i,j,k} = D_k, \quad \forall k, \ \forall j \in \CN_k. 
  \end{flalign*}
\end{enumerate}
\end{definition}

\begin{algorithm}[t]
\footnotesize
    \caption{Location-aware Content Delivery}
	\begin{algorithmic}[1]
		\Procedure{DELIVERY}{$\{\BQ_k\}$}
		
		\State $k^{*} = \arg\min_{k \in \CK} m_k$
            \State $\hat{\BQ} \gets \BQ_{k^{*}}$
            \State $\Bq \gets \mathbf{0}$
            \ForAll{$k \in [K]$}
            \State Form $\BG_k$ based on $\hat{\BQ}$
            \State Compute $\BP_k$ based on $\BG_k$ using Algorithm~\ref{alg:File_ind_gen}
            \EndFor
		\ForAll{$n \in [\hat{N}]$} %\label{alg:delivery_K_subsets}
		%\If{$|\CU(n) \setminus \CK_p| \geq L$} \label{alg:phantom_DoF_check}
		\State $\Bx(n) \gets 0$
		        \ForAll{$k \in \CU(n)$}
                    \State $x_{k} \gets 0$
		            \ForAll{$m \in [\alpha]$}
                        \State $p \gets \BP_{k}(n,m)$
                        \State $\Bq(W_{p}(s_k)) \gets \Bq(W_{p}(s_k)) +1$
                        \State $\chi \gets \Bq(W_{p}(s_k))$
		                \State $x_k \gets$ $ \prod\left(x_{k}, W_{p}^{\chi}(s_k)\right)$ 
		            \EndFor
              \State $\Bx(n) \gets \Bx(n) + \Bv_{k}(n)x_{k}$
		        \EndFor
		 %   \EndIf
            \State Transmit $\Bx(n)$
		\EndFor
		\EndProcedure
	\end{algorithmic}
	\label{Alg:Delivery_main}
\end{algorithm}

The conditions C1-C4 in Definition~\ref{FM_matrix_definition} constitute a system of equations that needs to be solved to obtain the matrix $\BG_k$. The details of creating and solving this system of equations are provided in Appendix~\ref{sec: Appendix A}. 

Once $\{\BG_k\}$ matrices are found, we can use Algorithm~\ref{alg:File_ind_gen} to build $\{\BP_k\}$. In a nutshell, in each round (indexed by $n \in [\hat{N}]$ and $m \in [\alpha]$), if $g_{\Check{f}_k^n, \omega, k}$ is positive, a file-fragment of $W_\omega(s_{k}), k \in \CU(n)$ is assigned to $\BP_k(n,m)$ and  $g_{\Check{f}_k^n, \omega, k}$ is subtracted by $1$ for the following rounds. The whole delivery process is summarized in Algorithm~\ref{Alg:Delivery_main}, and a clarifying example is provided in the following.

\begin{exmp}
\label{exmp:interference-free-delivery}
Consider the network in Example~\ref{exmp:virt_network}, with the set of requested subfiles for each user ($\CM_1$-$\CM_4$) given in~\eqref{eq:exmp_req_subfiles}. Following Definition~\ref{FM_matrix_definition},
the TIS (the set of temporary indices) of each user can be written as
%\begin{align*}
%    \Check{\CN}_{{1}} &= \{4, 5, 6, 7, 8, 9, 10, 11, 12\}, \\
%    \Check{\CN}_{{2}} &= \{1,2,3,7,8,9,10,11,12\}, \\
%    \Check{\CN}_{{3}} &= \{1,2,3,4,5,6,10,11,12\}, \\
%    \Check{\CN}_{{4}} &= \{1,2,3,4,5,6,7,8,9\},
%\end{align*}
$\Check{\CN}_{{1}} = \{4, 5, 6, 7, 8, 9, 10, 11, 12\}$, $\Check{\CN}_{{2}} = \{1,2,3,7,8,9,10,11,12\}$, $\Check{\CN}_{{3}} = \{1,2,3,4,5,6,10,11,12\}$, and $\Check{\CN}_{{4}} = \{1,2,3,4,5,6,7,8,9\}$, 
while the RIS (the set of missing subfile indices) for each user is given as $\CN_1 = \{4, 5, 6, 7, 8, 9, 10, 11, 12\}, \, \CN_2 = \{3,4\}, \, \CN_3 = \{4\}$ and $\CN_4 = \{1,2\}$.
Note that while the TIS is the same size for all users, the size of the RIS can be different. To map each RIS into its corresponding TIS, we first need to divide the requested subfiles of users $1$-$4$ into $D_1 =2$, $D_2={9}$, $D_3={18}$, and $D_4=9$ file-fragments, respectively (note that $\alpha = 2$ is required to have integer $D_k$ also for $k=2,4$). Then, according to~\eqref{eq:original data vector}, each subfile $W_{\Check{f}_k^n}(s_k)$ is replaced by the concatenation of $\alpha=2$ missing file-fragments $W^q_{\BP_k(n,m)}(s_k)$. As a result, recalling from Example~\ref{exmp:virt_network} that the size of each missing subfile for users $1$-$4$ is $\frac{1}{12}$, $\frac{1}{4}$, $\frac{1}{4}$, and $\frac{1}{4}$ data units, the size of the transmitted data to these users would be $\frac{1}{12}, \frac{1}{18}, \frac{1}{36}$ and $\frac{1}{18}$ data units, respectively. Note that the size of the intended data for each user is proportional to the approximated rate at its location.

Now, let us review how the transmission vector $\Bx(n)$ is built from PTV $\Check{\Bx}(n)$. First, we form the FM matrices $\BG_1$-$\BG_4$ to map each RIS $\CN_k$, $k \in [K]$ to its corresponding TIS $\Check{\CN}_{k}$. One such set of matrices is as follows
\begin{equation*} 
    \begin{aligned}
    &\BG_1 = \begin{psmallmatrix}
    0 & 0 & 0 & 0 & 0 & 0 & 0 & 0 & 0 & 0 & 0 & 0 \\
    0 & 0 & 0 & 0 & 0 & 0 & 0 & 0 & 0 & 0 & 0 & 0 \\
    0 & 0 & 0 & 0 & 0 & 0 & 0 & 0 & 0 & 0 & 0 & 0 \\
    0 & 0 & 0 & 2 & 0 & 0 & 0 & 0 & 0 & 0 & 0 & 0 \\
    0 & 0 & 0 & 0 & 2 & 0 & 0 & 0 & 0 & 0 & 0 & 0 \\
    0 & 0 & 0 & 0 & 0 & 2 & 0 & 0 & 0 & 0 & 0 & 0 \\
    0 & 0 & 0 & 0 & 0 & 0 & 2 & 0 & 0 & 0 & 0 & 0 \\
    0 & 0 & 0 & 0 & 0 & 0 & 0 & 2 & 0 & 0 & 0 & 0 \\
    0 & 0 & 0 & 0 & 0 & 0 & 0 & 0 & 2 & 0 & 0 & 0 \\
    0 & 0 & 0 & 0 & 0 & 0 & 0 & 0 & 0 & 2 & 0 & 0 \\
    0 & 0 & 0 & 0 & 0 & 0 & 0 & 0 & 0 & 0 & 2 & 0 \\
    0 & 0 & 0 & 0 & 0 & 0 & 0 & 0 & 0 & 0 & 0 & 2
    \end{psmallmatrix},
    \BG_2 = \begin{psmallmatrix}
    0 & 0 & 0 & 2 \\
    0 & 0 & 0 & 2 \\
    0 & 0 & 0 & 2 \\
    0 & 0 & 0 & 0 \\
    0 & 0 & 0 & 0 \\
    0 & 0 & 0 & 0 \\
    0 & 0 & 2 & 0 \\
    0 & 0 & 2 & 0 \\
    0 & 0 & 2 & 0 \\
    0 & 0 & 0 & 2 \\
    0 & 0 & 2 & 0 \\
    0 & 0 & 1 & 1 
    \end{psmallmatrix}, \\
    &\BG_3 = \begin{psmallmatrix}
    0 & 0 & 0 & 2 \\
    0 & 0 & 0 & 2 \\
    0 & 0 & 0 & 2 \\
    0 & 0 & 0 & 2 \\
    0 & 0 & 0 & 2 \\
    0 & 0 & 0 & 2 \\
    0 & 0 & 0 & 0 \\
    0 & 0 & 0 & 0 \\
    0 & 0 & 0 & 0 \\
    0 & 0 & 0 & 2 \\
    0 & 0 & 0 & 2 \\
    0 & 0 & 0 & 2
    \end{psmallmatrix}, \quad
    \BG_4 = \begin{psmallmatrix}
    2 & 0 & 0 & 0 \\
    2 & 0 & 0 & 0 \\
    2 & 0 & 0 & 0 \\
    0 & 2 & 0 & 0 \\
    1 & 1 & 0 & 0 \\
    2 & 0 & 0 & 0 \\
    0 & 2 & 0 & 0 \\
    0 & 2 & 0 & 0 \\
    0 & 2 & 0 & 0 \\
    0 & 0 & 0 & 0 \\
    0 & 0 & 0 & 0 \\
    0 & 0 & 0 & 0 
    \end{psmallmatrix}.
    \end{aligned}
\end{equation*}
Accordingly, based on FM matrices and Algorithm~\ref{alg:File_ind_gen}, the file-fragment index matrices $\BP_k$ are built as follows
\begin{equation*} 
    \begin{aligned}
    \BP_1 = \begin{psmallmatrix}
    4 & 4  \\
    7 & 7  \\
    10 & 10 \\
    0 & 0 \\
    5 & 5  \\
    6 & 6  \\
    8 & 8  \\
    9 & 9  \\
    0 & 0  \\
    11 & 11  \\
    0 & 0  \\
    12 & 12  
    \end{psmallmatrix},
    \BP_2 = \begin{psmallmatrix}
    4 & 4  \\
    0 & 0  \\
    4 & 4 \\
    3 & 3 \\
    3 & 4  \\
    4 & 4  \\
    0 & 0  \\
    3 & 3  \\
    3 & 3  \\
    4 & 4  \\
    3 & 3  \\
    0 & 0 
    \end{psmallmatrix},
    \BP_3 = \begin{psmallmatrix}
    4 & 4  \\
    4 & 4  \\
    0 & 0 \\
    4 & 4 \\
    0 & 0  \\
    4 & 4  \\
    4 & 4  \\
    4 & 4  \\
    4 & 4  \\
    0 & 0  \\
    4 & 4  \\
    4 & 4
    \end{psmallmatrix},
    \BP_4 = \begin{psmallmatrix}
    0 & 0  \\
    1 & 1  \\
    1 & 1 \\
    2 & 2 \\
    2 & 1  \\
    0 & 0  \\
    2 & 2  \\
    0 & 0  \\
    2 & 2  \\
    1 & 1  \\
    1 & 1  \\
    2 & 2
    \end{psmallmatrix}.
    \end{aligned}
\end{equation*}
Therefore, using~\eqref{eq:original data vector}, the first transmission vector would be built as
%\begin{equation*}%\label{eq:first_transmission} 
    $\Bx(1) = \Bv_{1}(1) \prod(A_{4}^{1}, A_{4}^{2}) +  \Bv_{2}(1) \prod(B_{4}^{1}, B_{4}^{2}) +  \Bv_{3}(1)\prod(C_{4}^{1}, C_{4}^{2}).$
%\end{equation*}
Note that $\BP_k(n,m) = 0$ means there is no transmission for user $k \in [K]$ in time-slot $n\in \hat{N}$. 

Then, the corresponding received signal at three users are 
\begin{equation} \label{eq:actual received signals_exmp}%\scriptsize
\begin{aligned}
 y_1 &= x_1\Bh_1^H\Bv_{1}(1)  + \underline{x_2}\Bh_1^H\Bv_{2}(1) 
 +\underline{x_3}\Bh_1^H\Bv_{3}(1) + z_1, \\  y_2 &= \underline{x_1} \Bh_2^H \Bv_{1}(1) + x_2 \Bh_2^H \Bv_{2}(1) + \underline{\underline{x_3}}\Bh_2^H \Bv_{3}(1)+z_2 \; , \\
       y_3 &= \underline{\underline{x_1}}\Bh_3^H\Bv_{1}(1) +  \underline{\underline{x_2}}\Bh_3^H\Bv_{2}(1)+ x_3 \Bh_3^H \Bv_{3}(1) + z_3 \; ,
\end{aligned}
\end{equation}
where, $x_1 = \prod(A_{4}^{1}, A_{4}^{2}), x_2 = \prod(B_{4}^{1}, B_{4}^{2})$ and $x_3 = \prod(C_{4}^{1}, C_{4}^{2})$. Recalling that files $A$, $B$, and $C$ correspond to the content of states $1$, $2$, and $3$, respectively, based on the MLPDAs $\BQ_1$-$\BQ_3$, the underline terms in the received signal $y_k$ are available the user $k$'th cache memory. Hence, by estimating $\Bh_k^{H} \Bv_{j}(1)$ (using precoded downlink pilots), the underlined interference terms of the received signals can be removed. Furthermore, the terms with double underline in equation~\eqref{eq:actual received signals_exmp} are suppressed through the utilization of beamforming vectors $\Bv_1(1)$, $\Bv_2(1)$, and $\Bv_3(1)$, which are designed based on the user' interference sets $\CI_1(1)$ to $\CI_3(1)$ (cf. Example~\ref{exmp:virt_network}). Hence, the intended data can be successfully decoded by each user. Note that user~$1$ has to receive more data compared to users~$2$ and~$3$ due to its larger file size, which is in line with its higher achievable data rate as it is located in state $s_1=1$ (cf. Example~\ref{exmp:placement}). %be removed from the received signal and $\prod(A_{4}^{1}, A_{4}^{2})$ can be decoded interference-free. Decoding at users~$2$ and~$3$ also follows a similar procedure. 
%
% Moreover, the received signal at users $2$ and $3$ are as follows
% \begin{equation*}
% \begin{aligned}
%        y_2 &=  \prod(B_{4}^{1}, B_{4}^{2}) \Bh_2^H \Bv_{2}(1) + \underline{\prod(A_{4}^{1}, A_{4}^{2})} \Bh_2^H \Bv_{1}(1) + w_2 \; , \\
%        y_3 &= \prod(C_{4}^{1}, C_{4}^{2}) \Bh_3^H \Bv_{3}(1) + w_3 \; , 
% \end{aligned}
% \end{equation*}
% where $w_2 \equiv \Bh_2^H \Bv_{3}(1)\prod(C_{4}^{1}, C_{4}^{2})+z_2$ and $w_3 \equiv \Bh_3^H(\Bv_{1}(1) \prod(A_{4}^{1}, A_{4}^{2}) +  \Bv_{2}(1) \prod(B_{4}^{1}, B_{4}^{2}))+z_3$ contain the (suppressed) interference terms plus noise at users $2$ and $3$, respectively.
% Following a similar argument as above, users $2$ and $3$ can decode their intended data interference-free. The rest of the transmitted messages can be created following similar steps. Note that compared to conventional works~\cite{MLPDA_ISIT2022,lampiris2018adding,salehi2020lowcomplexity}, in all these transmissions, we allocate higher data transmission rates to user $1$ ($1.5$ times) and lower rates to user $3$ (half a time) than those assigned to users $2$ and $4$.
\end{exmp}

\subsection{Weighted Max-Min Beamforming} \label{Sec:beamforming}

In this section, we explain the optimal design of precoding vectors $\Bv_{k}(n)$ in~\eqref{eq: multiplexed signal} to enhance the finite-SNR performance. Unlike the conventional max-min beamforming problem addressed in previous works such as~\cite{tolli2017multi, salehi2020lowcomplexity, D2D_CC_mahmoodi_2022}, we consider a weighted-max-min (WMM) beamforming approach to enable multi-rate transmission. In this regard, we alter the iterative approach presented in~\cite{salehi2020lowcomplexity}, with a slight modification to account for the weighted beamforming requirements, wherein the weights reflect the non-uniform quantities of data transmitted to different users. To this end, we first briefly review the delivery process.

As discussed in Section~\ref{sec:delivery}, the proposed scheme serves all the users using $\hat{N}$ transmission vectors $\Bx(n)$. Every vector $\Bx(n)$ comprises $|\CU(n)|$ data terms $x_{k}$ and the same number of beamforming vectors $\Bv_{k}(n)$ as represented in~\eqref{eq:recieved_signal_sysmodel}. Hence, the corresponding received signal at user $k$ in $n$'th transmission can be rewritten as
\begin{equation} \label{eq:received_signal} 
    y_{k} = \Bh_{k}^{H}\Bv_{k}(n)x_{k} + \hspace{-2mm}\sum_{i \in \overline{\CI_k}(n)}\hspace{-2mm}\Bh_{k}^{H}\Bv_{i}(n)x_{i} + \hspace{-6mm}\sum_{j \in \CU(n) \setminus \overline{\CI_k}(n)}\hspace{-3mm}\underline{\Bh_{k}^{H}\Bv_{j}(n)x_{j}}+ z_k,
\end{equation}
where $\overline{\CI_k}(n) = \{i \in \CU(n) \ | \ k \in \CI_i(n)\}$. In~\eqref{eq:received_signal}, each intended message $x_k$ contains a fresh data  for user $k$ with size 
\begin{equation} \label{C_k}
    c_k \equiv \frac{\alpha }{ F_{s_k} \alpha  (\frac{\hat{F}-\hat{Z}}{F_{s_k}-Z_{s_k}})} = \frac{1-m_k}{\hat{F}-\hat{Z}}.
\end{equation}
Recall that $x_{k}$ comprises $\alpha$ file-segments, where each segment is $\frac{1}{D_k}$ of a subfile and each subfile is $\frac{1}{F_{s_k}}$ data units in size. Note that the underlined terms in~\eqref{eq:received_signal} are removed from the received signals by utilizing cache memories and estimating the equivalent channels $\Bh_{k}^{H}\Bv_{j}(n)$ prior to the decoding process. Consequently, these terms are not regarded as interference. So, the received SINR at user $k$ is calculated as follows
\begin{equation}\label{eq:received_SINR}
    \gamma_k = \frac{|\Bh_{k}^{H}\Bv_{k}(n)|^2}{\underset{{{i \in \overline{\CI_k}(n)}}}{\sum}|\Bh_{k}^{H}\Bv_{i}(n)|^2+ N_0} .
\end{equation}
Moreover, the required time to deliver $x_k$ is calculated as $T_k = \frac{c_k}{\log (1 + \gamma_k)}$. Thus, the total delivery time to send $\Bx(n)$ is $T_{\CU(n)} = \max_{k \in \CU(n)}T_k$ [seconds]. 

Now, since we aim to minimize the delivery time, the beamformer optimization problem is formulated as $\min_{\{\Bv_{k}\}} \max_{k \in \CU(n)} T_{k}$, or equivalently as $\max_{\{\Bv_{k}\}} \min_{k \in \CU(n)} \frac{1}{c_k}\log(1 + \gamma_k)$, where $\log(1 + \gamma_k)$ is the dedicated rate to user $k$. Hence, the weighted minimum rate maximization for a given transmission $\Bx_{\CU(n)}$ can be formulated as $\max_{\{\Bv_{k}(n), \gamma_k\}} \min_{k \in \CU(n)}  \frac{1}{c_k}\log(1 + \gamma_k)$. Now, an equivalent objective can be achieved by using $\max_{\{\Bv_{k}(n), \bar{\gamma}, \gamma_k\}} \bar{\gamma}$, such that $\bar{\gamma} \leq \frac{1}{c_k}\log(1 + \gamma_k), \forall k \in \CU(n)$. By utilizing the fact that $\bar{\gamma} \leq \frac{1}{c_k}\log(1 + \gamma_k)$ is equivalent to $\exp(\bar{\gamma}) \leq (1+\gamma_k)^{\frac{1}{c_k}}$, we can express the weighted-max-min (WMM) beamforming problem as:
\begin{equation} \label{opt:nonconvex_WMMF_epic_MJ}
\begin{array} {l}
		\underset{\{\Bv_{k}(n), {\gamma}, \gamma_k\}}\max \ {\gamma}  \\
	 \textrm{s.t.}  \\ \quad
  \gamma^{c_{k}}-1 \leq \gamma_k, \quad \forall k \in \CU(n), \\ \quad
  \gamma_k \leq \frac{|\Bh_{k}^{H}\Bv_{k}(n)|^2}{\underset{{\substack{i \in \overline{\CI_k}(n)}}}{\sum}|\Bh_{k}^{H}\Bv_{i}(n)|^2+ N_0}, \quad \forall k \in \CU(n), \\ \quad
	  \sum_{k \in \CU(n)}\|\Bv_{k}(n)\|^2\leq P_T, 
\end{array}
\end{equation} 
where $P_T$ is the transmit power and $\gamma = \exp{(\bar{\gamma})}$. The quasi-convex problem~\eqref{opt:nonconvex_WMMF_epic_MJ} is similar to the one discussed in~\cite{salehi2020lowcomplexity}, but with additional convex constraints $\gamma^{c_{k}}-1 \leq \gamma_k, \ \forall k \in \CU(n)$.~\footnote{Note that $\gamma^{c_{k}}$ is convex for $c_k \geq 1$. Since $\{c_k\}$ are just considered as weights for rate allocation purposes, they can be easily replaced by $\{c_k \gets {c_k}/{\min_{k \in [K]} c_k}\}$ to make $\gamma^{c_{k}}$ a convex function.} This problem can be optimally solved by conducting a search over $\gamma$ using bisection and applying the Lagrangian duality (LD) scheme in~\cite{salehi2020lowcomplexity} for a fixed $\gamma_k=\gamma^{c_k}-1$.

The LD scheme is an iterative fast-beamforming method used for linear-beamformer design in the literature (c.f.,~\cite{salehi2020lowcomplexity} and the references therein). Since it has already been thoroughly described in~\cite{salehi2020lowcomplexity}, we will briefly review the LD scheme and its modifications for our WMM optimization. Specifically, we iteratively determine the optimal value of $\gamma$ using a bisection search, where $\gamma_k = \gamma^{c_k}-1$. Once $\gamma$ is fixed, we employ the fixed point iteration method~\cite{salehi2020lowcomplexity} to obtain the dual variables $\nu_k$ for $\forall k \in \CU(n)$. Specifically, in the initial step, we initialize $\nu_k$ as $\nu_k \gets \nu_k[1]$. Subsequently, we iteratively update the dual variables until the desired level of convergence is attained. For each iteration $\tau+1$, the dual variable $\nu_k$ is updated according to the following:
\begin{equation}\label{eq:fual_powers}
    \nu_k (\tau+1) = (\gamma^{c_k}-1)\left({\Bh^{H}_k \mathbf{\Sigma}^{-1}_{k}(\tau) \Bh_k}\right)^{-1},
\end{equation}
where $\mathbf{\Sigma}_{k}(\tau) = {\sum}_{{\hat{k} \in \CI_k(n) }}\nu_{\hat{k}}(\tau) \Bh_{\hat{k}}\Bh^{H}_{\hat{k}}+N_0\BI$. When~\eqref{eq:fual_powers} is converged, the normalized beamformer $\bar{\Bv}_{k}(n), \forall k \in \CU(n)$, is calculated as $\bar{\Bv}_{k}(n)=\mathbf{\Sigma}^{-1}_{k} \Bh_k/\|\mathbf{\Sigma}^{-1}_{k}\Bh_k\|$, where $\mathbf{\Sigma}_{k} = {\sum}_{{\hat{k} \in \CI_k(n) }}\nu_{\hat{k}} \Bh_{\hat{k}}\Bh^{H}_{\hat{k}}+N_0\BI$. To determine the power vector of the beamformers, denoted by $\tilde{\mathbf{p}}$, we can follow the same steps as presented in~\cite[Eq.(26)]{salehi2020lowcomplexity}. Note that $\gamma$ will be updated based on $\tilde{\mathbf{p}}$ to satisfy the power constraint $P_T$. Hence, the set of optimal downlink beamformers is computed as $\Bv_{k}(n) = \sqrt{p_k}\bar{\Bv}_{k}(n)$, where $p_k$ is the $k$'th entry of $\tilde{\mathbf{p}}$.
\begin{remark}\label{remark_proportional_rate}
    Similar to~\cite{mahmoodi2022asymmetric_arxiv_TWC2023}, the proposed WMM beamforming in~\eqref{opt:nonconvex_WMMF_epic_MJ} results in proportional rate allocation such that $\frac{\log(1+\gamma_k)}{\log(1+\gamma_{\bar{k}})} = \frac{c_{\bar{k}}}{c_{k}}$, $\forall(k,\bar{k}) \in \CU(n)$ . However, unlike~\cite{mahmoodi2022asymmetric_arxiv_TWC2023}, the proposed delivery scheme in this paper removes the need for successive interference cancellation (SIC) at the receiver. Herein, the SIC requirement is removed due to condition C3 in Definition~\ref{MLPDA_definition_new}, which prevents multiple message transmissions to a single user at a time. This paper proposes a scheme that is more suitable for large networks, not only because of its simplified transceiver design, but also because of its reduced processing requirements for signal-level delivery, decreased subpacketization~\cite{salehi2021lowsubpacketization}, and the ability to employ shared caching concepts~\cite{Shahred_cache_Emanuele_2020,Shared_cache_decentrlized_2021,abolpour2023cache}.
\end{remark}

Now, using Eq.~\eqref{C_k}, the total transmission time $T_T$ for the whole delivery process can be calculated as
\begin{equation}\label{eq:T_T_exact}
    T_T = \frac{1}{\hat{F}-\hat{Z}}\sum_{n \in [\hat{N}]}\max_{k \in \CU(n)}\frac{1-m_k}{\log(1+\gamma_k)}.
\end{equation}
It is imperative to note that the calculation of the total transmission time $T_T$ in equation~\eqref{eq:T_T_exact}, relies on knowing the beamformers for all $\hat{N}$ transmissions, which in turn requires accurate information about user locations and channel conditions. However, during the placement phase, the actual user locations and channel states are unknown, rendering $T_T$ computation infeasible. To address this challenge, we adopt an approximation for $T_T$ for placement purposes that assumes a uniform access probability for all STUs. This approximation is in accordance with the findings presented in~\cite{mahmoodi2022asymmetric_arxiv_TWC2023} and is based on the delivery in Section~\ref{sec:delivery}. Thus, the same approximation as in~\cite{mahmoodi2022asymmetric_arxiv_TWC2023} is obtained in this paper, despite the utilization of different delivery and placement techniques.
\begin{lemma}\label{lemm:aprox_delivery_time}
The total delivery time $T_T$ calculated in~\eqref{eq:T_T_exact} can be approximated as
\begin{equation}\label{eq:T_T_approx}
    \hat{T}_T = \frac{K}{K\bar{m}+L}\max_{s \in \CS}\frac{1-m(s)}{r(s)} \; .
\end{equation}
\end{lemma}
\begin{proof}
We first substitute $log(1+\gamma_k)$ in~\eqref{eq:T_T_exact} with its upper bound $r(s_k)$ given in~\eqref{eq: state-rate} to get
\begin{equation}\label{eq:T_T-first_approx}
    T_T \sim \frac{1}{\hat{F}-\hat{Z}}\sum_{n \in [\hat{N}]}\max_{k \in \CU(n)}\frac{1-m_k}{r(s_k)} \; .
\end{equation}
Then, using inequality $\max\limits_{k \in {\CU(n)}}\frac{1-m(s_k)}{r(s_k)} \leq \max\limits_{s \in \CS}\frac{1-m(s)}{r(s)}$, we substitute the RHS of~\eqref{eq:T_T-first_approx} with its upper bound to get $T_T \sim 
\frac{\hat{N}}{\hat{F}-\hat{Z}}\max\limits_{s \in \CS}\frac{1-m(s)}{r(s)}$. Note that $\frac{\hat{N}}{\hat{F}-\hat{Z}} = \frac{K}{\frac{K(\hat{F}-\hat{Z})}{\hat{N}}}$, where $\frac{K(\hat{F}-\hat{Z})}{\hat{N}}$ is nothing but the sum-DoF~\cite{MLPDA_ISIT2022}. Hence, approximating $\frac{K(\hat{F}-\hat{Z})}{\hat{N}}$ with its upper bound $K\hat{m}+L$~\cite{MLPDA_ISIT2022}, we have 
\begin{equation}\label{T_T_approx_hat}
    T_T \sim \frac{K}{K\hat{m}+L}\max_{s \in \CS}\frac{1-m(s)}{r(s)} \; .
\end{equation}
Note that $\hat{m} = \min_{k \in [K]}m(s_k)$ requires user location knowledge, which is unknown during the placement phase. Thus, replacing $\hat{m}$ with its lower bound $\bar{m}$, \eqref{eq:T_T_approx} is achieved.
\end{proof}

\section{Simulation Results}
\label{sec:Simulations}
To evaluate the proposed location-dependent scheme, we conduct numerical simulations in a scenario similar to the one studied in~\cite{mahmoodi2022asymmetric_arxiv_TWC2023}, albeit with a much larger number of users $K$.\footnote{
The number of users in~\cite{mahmoodi2022asymmetric_arxiv_TWC2023} is limited to nine due to the complex transceiver design and the large subpacketization requirement. Of course, it provides an improved multicasting gain by serving multiple users with a single message.}
%The method in~\cite{mahmoodi2022asymmetric_arxiv_TWC2023} provides multicasting gain by serving multiple users with a single message. However, its scalability is limited due to the complex transceiver design and large subpacketization requirements, making it more suitable for small networks. Hence, the simulation results in~\cite{mahmoodi2022asymmetric_arxiv_TWC2023} consider a maximum of $K = 9$ users.
%}
Specifically, we consider a $30 \times 30 \mathrm{[m^2]}$ XR application environment, where a
%conducted in a $30 \times 30 \mathrm{[m^2]}$ room is examined. A 
unique 3D image is required at each STU of size $1 \times 1 \mathrm{[m^2]}$ to reconstruct a detailed FoV (resulting in a total number of $S=900$ STUs).
%(i.e., $S = 900$). 
A transmitter with $L$ antennas is located on the ceiling, $5 \mathrm{[m]}$ above the floor, in the middle of the room.
%, on the ceiling, $5 \mathrm{[m]}$ above the floor, to serve the users with the requested data. 
We assume that the small-scale fading of the channel vectors $\Bh_k$ follows a Rayleigh distribution and use the path loss model of~\cite{mahmoodi2022asymmetric_arxiv_TWC2023} for a user at STU $s \in [S]$: 
\begin{equation*}
    PL(s) = 32.4[dB]+20\log_{10}(f) + 10\eta\log_{10}(d_s) + \zeta,
\end{equation*}
where $d_s$ represents the distance between the center of STU $s$ and the transmitter, $\eta = 3$ is the pass-loss exponent, and $f$ denotes the frequency. To simulate the effect of randomly placed objects that can obstruct the propagation path between the transmitter and receivers, we use the term $\zeta \sim \mathbb{N}(0, {\sigma})$, where $\sigma$ is the standard deviation. Note that $\zeta$ is similar to the shadowing effect observed in outdoor propagation environments. We calculate the expected delivery rates $r(s)$ for the initial memory allocation in~\eqref{cache-allocation} by averaging over the rate values for all possible user locations and channel realizations in a given STU (c.f., Eq.~\eqref{eq: state-rate}). Without considering the shadowing effect $\zeta$, the transmit power is assumed to provide a $5 \mathrm{[dB]}$ SNR at the room boundaries, (unless mentioned otherwise). During the delivery phase, optimized weighted max-min beamformers~\eqref{opt:nonconvex_WMMF_epic_MJ} are used, and users are assumed to be located at any STU with equal probability.

Due to the stringent requirement of XR applications for a bounded delivery time, we use the 95-percentile of expected delivery time as a key performance metric (see~\cite{mahmoodi2022asymmetric_arxiv_TWC2023}).
%
%The importance of delivery time in XR applications highlights the need for the 95-percentile of expected delivery time as a key performance metric in this study (see~\cite{mahmoodi2022asymmetric_arxiv_TWC2023}). In our analysis, we 
We evaluate four distinct placement and delivery schemes:
%, in all of which we utilize the shared caching approach proposed in~\cite{abolpour2023cache} to construct their proper PDA, satisfying the conditions outlined in Definition~\ref{MLPDA_definition_new} for all schemes (and Definition~\ref{LAPDA_definition} for location-dependent schemes):
\begin{itemize}
    \item \textbf{Single-user placement, w/o CC:}
    This scheme employs a memory allocation process similar to~\cite{Mahmoodi_immersive_isit2021} that maximizes the local caching gain at bottleneck areas (equivalent to~\eqref{cache-allocation}, when the denominator of the objective function is ignored). Data transmission is done with conventional unicast beamforming and without using any CC technique.  
    \item \textbf{Multi-user placement, w/o CC}: This scheme employs the memory allocation process proposed in Section~\ref{sec:cache_placement}, but data transmission is done with conventional unicast beamforming and without using any CC technique.
    %and uses conventional unicast transmission without CC technique to transmit data.
    \item \textbf{Multi-user placement, w/ CC}: This scheme employs the memory allocation process proposed in Section~\ref{sec:cache_placement}, together with the location-dependent CC delivery scheme in Algorithm~\ref{Alg:Delivery_main};
    \item \textbf{Uniform placement, w/ CC}: This scheme employs the conventional uniform memory allocation together with the CC delivery scheme proposed in~\cite{abolpour2023cache} based on shared caching idea.% (for instance, c.f.,~\cite{MLPDA_ISIT2022,lampiris2018adding,salehi2020lowcomplexity});
\end{itemize}
In all of these schemes, we use the shared caching approach described in~\cite{abolpour2023cache} to construct a proper PDA, satisfying the conditions outlined in Definition~\ref{MLPDA_definition_new} (as well as Definition~\ref{LAPDA_definition} for location-dependent CC schemes).  

All schemes are compared in terms of their delivery times for $500$ random user drops and the resulting cumulative distribution functions (CDF) are shown in Fig.~\ref{fig:sim_cdf}. As shown, the scheme with uniform placement has the worst variation in total delivery time, making it unsuitable for applications with real-time content requests (e.g., XR gaming). This variation occurs because uniform placement only maximizes the minimum global caching gain, leading to optimal performance when all users have good connectivity but poor performance when some users experience poor connectivity in bottleneck areas. A location-dependent placement (even without any CC technique) can avoid this issue and improve overall performance. Furthermore, we observe that when no CC technique is used, single-user placement outperforms multi-user placement due to its higher local-caching gain. 
However, single-user placement is unsuitable for multicast CC transmission as it does not allocate any cache to the content requested in locations with good connectivity, reducing the possibility of achieving any coded caching gain. 
Finally, the proposed CC-based transmission scheme with multi-user placement provides the best performance as it enables a global caching gain while also avoiding wireless connectivity bottleneck areas.

%further improves the performance compared to single-user placement due to the higher achieved DoF.   

% \begin{figure}
%     \centering 
%     \includegraphics[width=0.7\columnwidth,keepaspectratio]{figs/Milad/Time_cdf_S7.eps}
%     \caption{The cumulative distribution function (CDF) of total delivery time (logarithmic-scale) for $500$ realizations, where $K = 36, M/S = 0.33, L=6$ and $\sigma = 5$.}
%     \label{fig:sim_cdf}
% \end{figure}

\begin{figure*}[ht]
\begin{minipage}[c]{0.32\textwidth}
\centering 
\setlength\abovecaptionskip{-0.25\baselineskip}
\includegraphics[width=1\columnwidth,keepaspectratio]{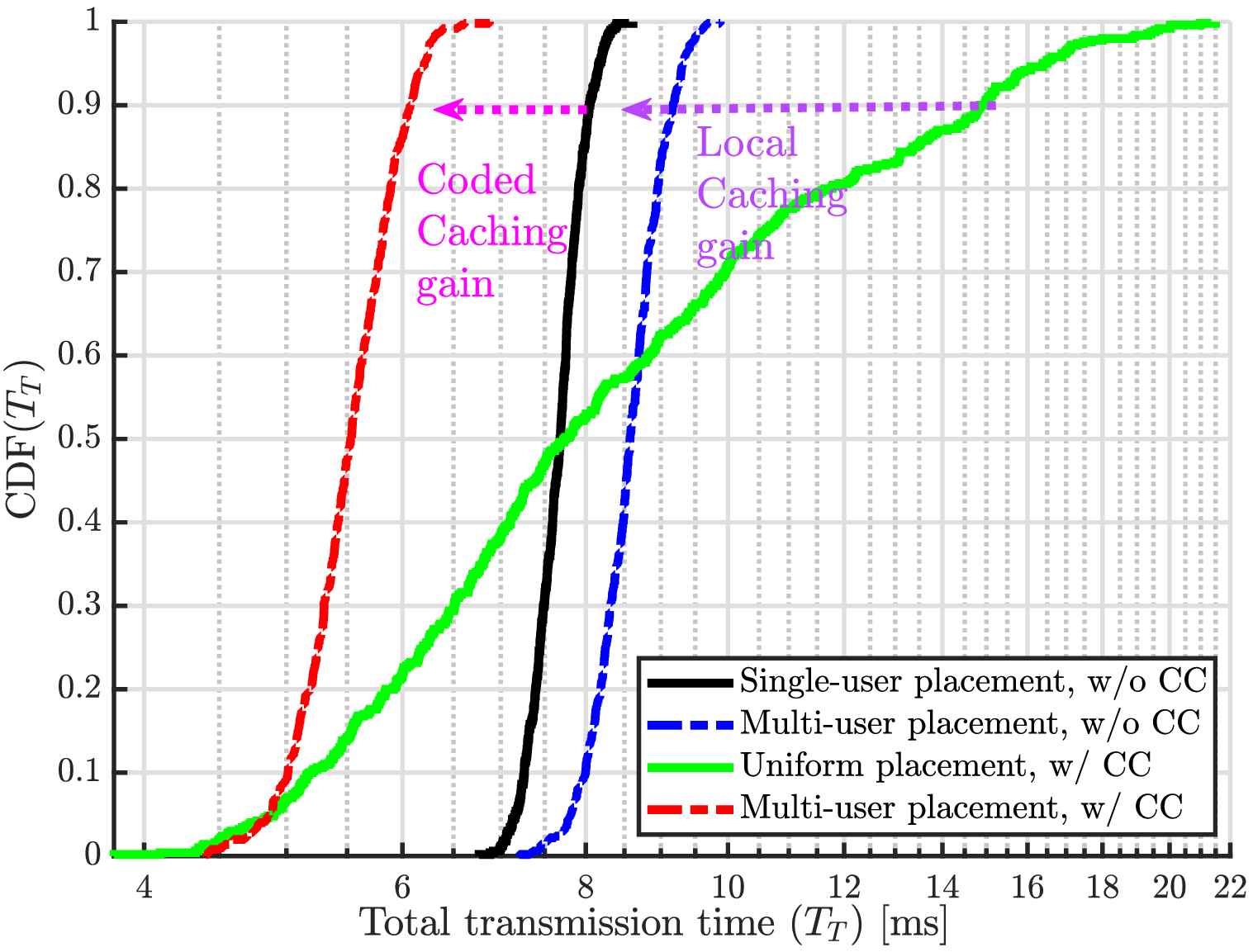}
\caption{The CDF of delivery time (logarithmic-scale) for $K = 36, M/S = 0.33, L=6$ and $\sigma = 7$.}
\label{fig:sim_cdf}
\end{minipage}
\hspace{1mm}
\begin{minipage}[c]{0.32\textwidth}
\centering 
\setlength\abovecaptionskip{-0.25\baselineskip}
\includegraphics[width=1\columnwidth,keepaspectratio]{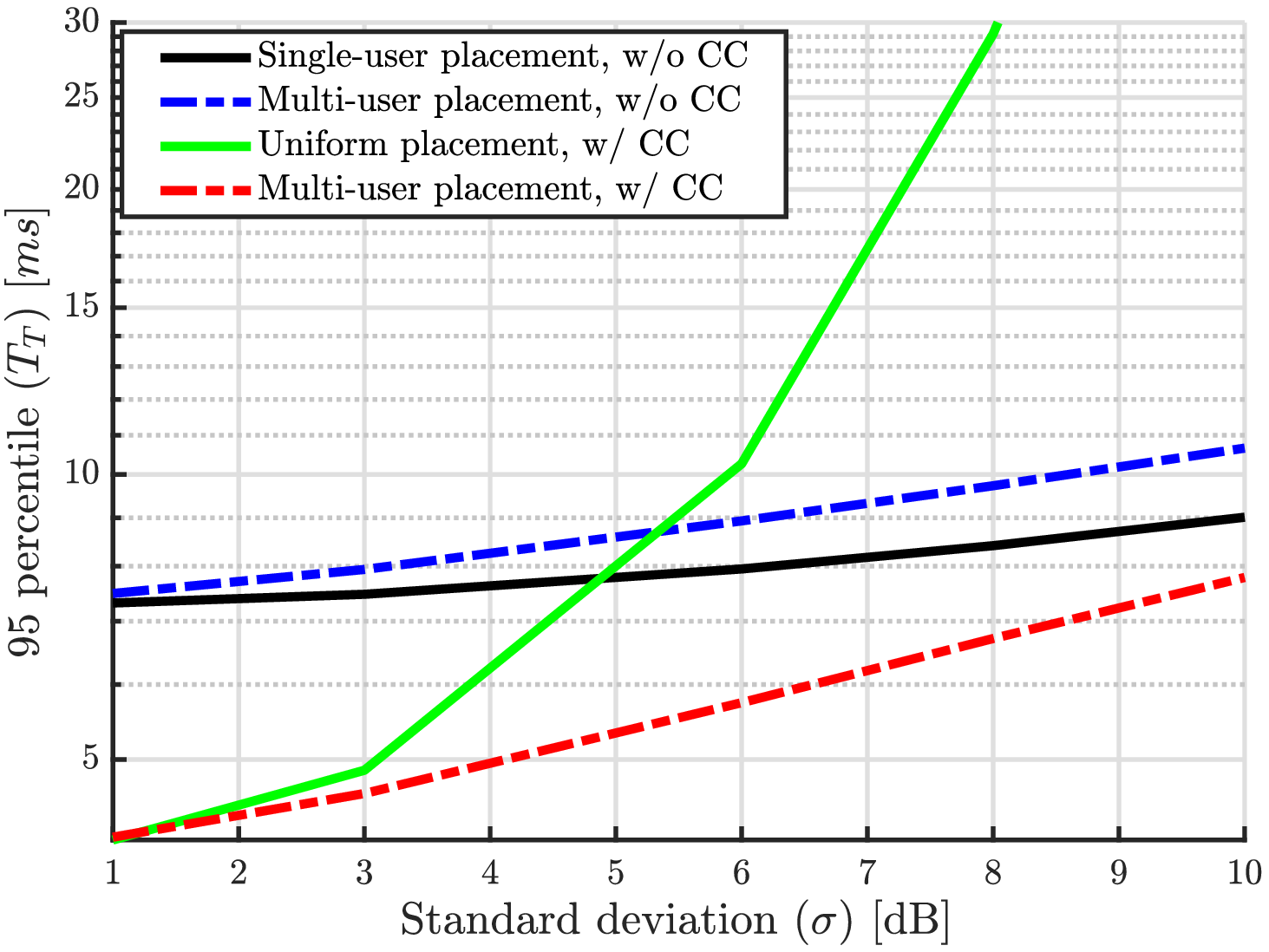}
\caption{Delivery time (logarithmic-scale) versus $\sigma$, where $K = 36, M/S = 0.33$, and $L=6$.}
\label{fig:sim_sigma_95time}
\end{minipage}
\hspace{1mm}
\begin{minipage}[c]{0.32\textwidth} 
\centering 
\setlength\abovecaptionskip{-0.25\baselineskip}
\includegraphics[width=1\columnwidth,keepaspectratio]{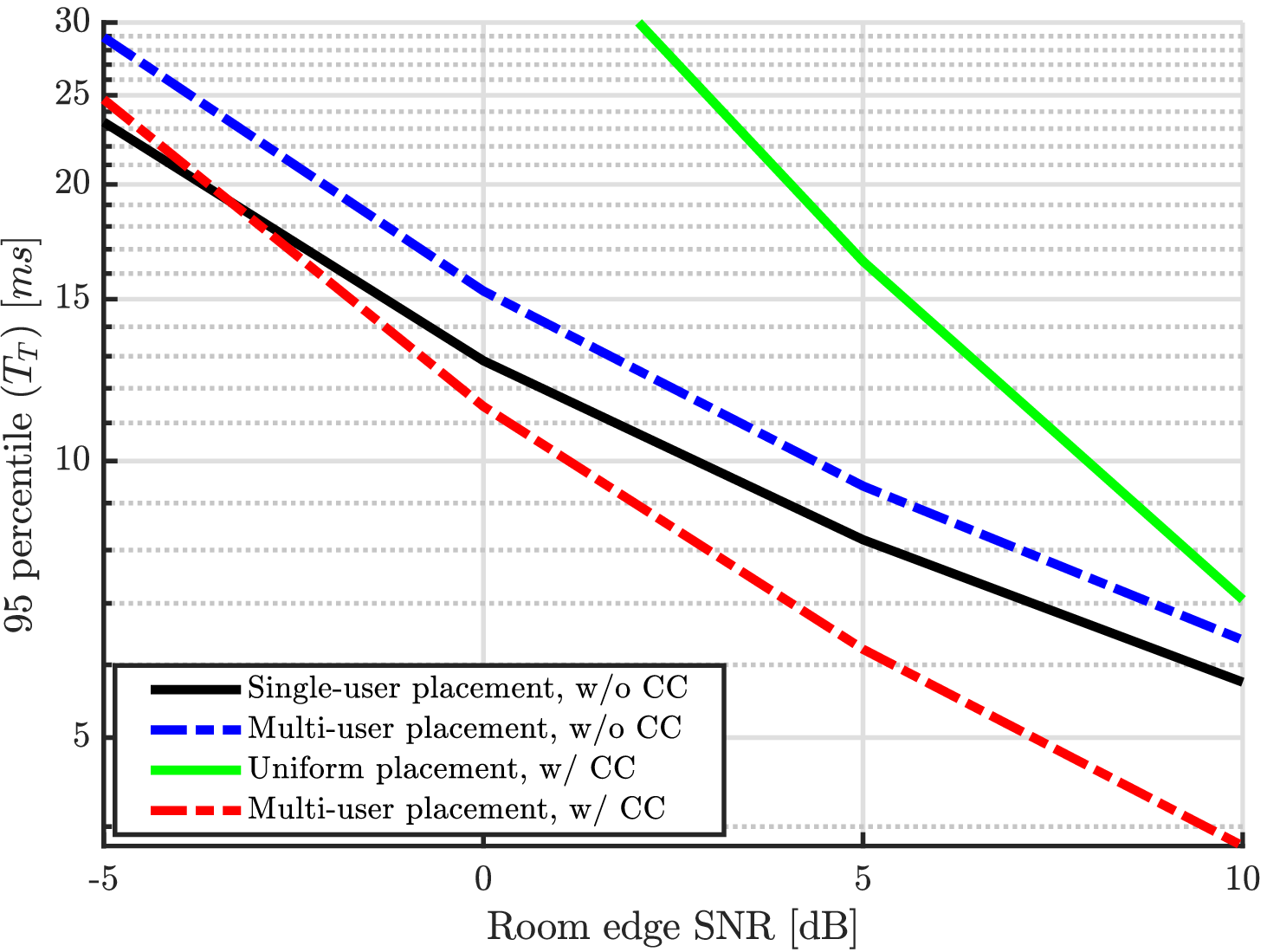}
\caption{Delivery time (logarithmic-scale) versus different room edge SNR,  where $K = 36, M/S = 0.33, \sigma=7$, and $L=6$.}
\label{fig:sim_power_95time}
\end{minipage}
\hspace{2mm}
\end{figure*}

Figure~\ref{fig:sim_sigma_95time} compares the performance of different schemes for various values of $\sigma$, which controls the attenuation intensity in different STUs. For small $\sigma$, the traditional uniform-placement method performs just as well as the proposed CC scheme with multi-user placement, and sometimes even better.
This is because, with small $\sigma$, the variation in large-scale fading among STUs is small, and hence, 
%, due to less variation in large-scale fading among STUs. In such a case, 
non-uniform memory allocation is unnecessary since it reduces the minimum achievable coded caching gain. However, the proposed scheme is more effective (outperforming all other schemes) for larger $\sigma$, as there are more attenuated STUs with significant rate differences to well-conditioned STUs. In fact, in these cases, the rate improvement for individual users outweighs the DoF loss caused by the memory allocation process (${K\hat{m}+L}$ vs ${K\frac{M}{S}+L}$ for the uniform placement case). %Hence, the proposed scheme outperforms the rest for large $\sigma$. In this case, the gains from location-dependent cache placement are significant, making the conventional CC schemes even worse than location-dependent unicasting schemes despite their larger achievable DoF. 

% \begin{figure*}[tb]
% \begin{minipage}[c]{0.49\textwidth}
% \centering 
% \setlength\abovecaptionskip{-0.25\baselineskip}
% \includegraphics[width=1\columnwidth,keepaspectratio]{figs/Milad/S_time.eps}
% \caption{Delivery time (logarithmic-scale) versus the standard deviation ($\sigma$), where $K = 36, M/S = 0.33$, and $L=6$.}
% \label{fig:sim_sigma_time}
% \end{minipage}
% \hspace{2mm}
% \begin{minipage}[c]{0.49\textwidth}
% \centering 
% \setlength\abovecaptionskip{-0.25\baselineskip}
% \includegraphics[width=1\columnwidth,keepaspectratio]{figs/Milad/S_time95.eps}
% \caption{Delivery time (logarithmic-scale) versus $\sigma$, where $K = 36, M/S = 0.33$, and $L=6$.}
% \label{fig:sim_sigma_95time}
% \end{minipage}
% \end{figure*} 

Figure~\ref{fig:sim_power_95time} compares different scheme based on the SNR value at the room border. As depicted, the single-user cache placement scheme is the best option when the received SNR is very low. In this case, the achievable rate at different locations is highly diverse, making the local caching gain the most influential factor in reducing the overall delivery time. Conversely, when the transmit power is high enough to make all locations have similar achievable rates, coded caching gain is the primary factor in reducing the overall transmission time. Therefore, uniform placement is optimal as it maximizes the minimum achievable DoF. 

Finally, Figure~\ref{fig:sim_L_95time} compares different schemes for a different number of transmit antennas ($L$). Results show that when $L$ is large, the coded caching gain is less effective in improving overall performance since $L$ is the main contributor to the achievable DoF, i.e., $K\hat{m}+L$. Thus, location-dependent schemes with no CC techniques perform almost as well as the proposed method with multicast CC transmission. However, when $L$ is relatively small, the coded caching gain is crucial in reducing the transmission time, and the proposed scheme is much more effective than the single-user cache placement case. Figures~\ref{fig:sim_M_95time} and~\ref{fig:sim_K_95time} support a similar conclusion: Figure~\ref{fig:sim_M_95time} illustrates that the local caching gain is the most influential factor in reducing delivery time when the available memory is small. This is because the CC gain $K\hat{m}$ is much less than the number of transmit antennas ($L$) in such scenarios. 
%However, when the available memory is large, the CC gain $K\hat{m}$ is considerable, becoming comparable to L. 
%Note that as $M$ increases, the schemes with single-user and multi-user cache placement and no CC transmissions both converge to the uniform placement case. 
Figure~\ref{fig:sim_K_95time} shows that the performance gap between the proposed method and the rest widens as the number of users increases due to higher achievable CC gain ($K\hat{m}$) for larger $K$.

\begin{figure*}[tb]
\begin{minipage}[c]{0.32\textwidth} 
\centering 
\setlength\abovecaptionskip{-0.25\baselineskip}\includegraphics[width=1\columnwidth,keepaspectratio]{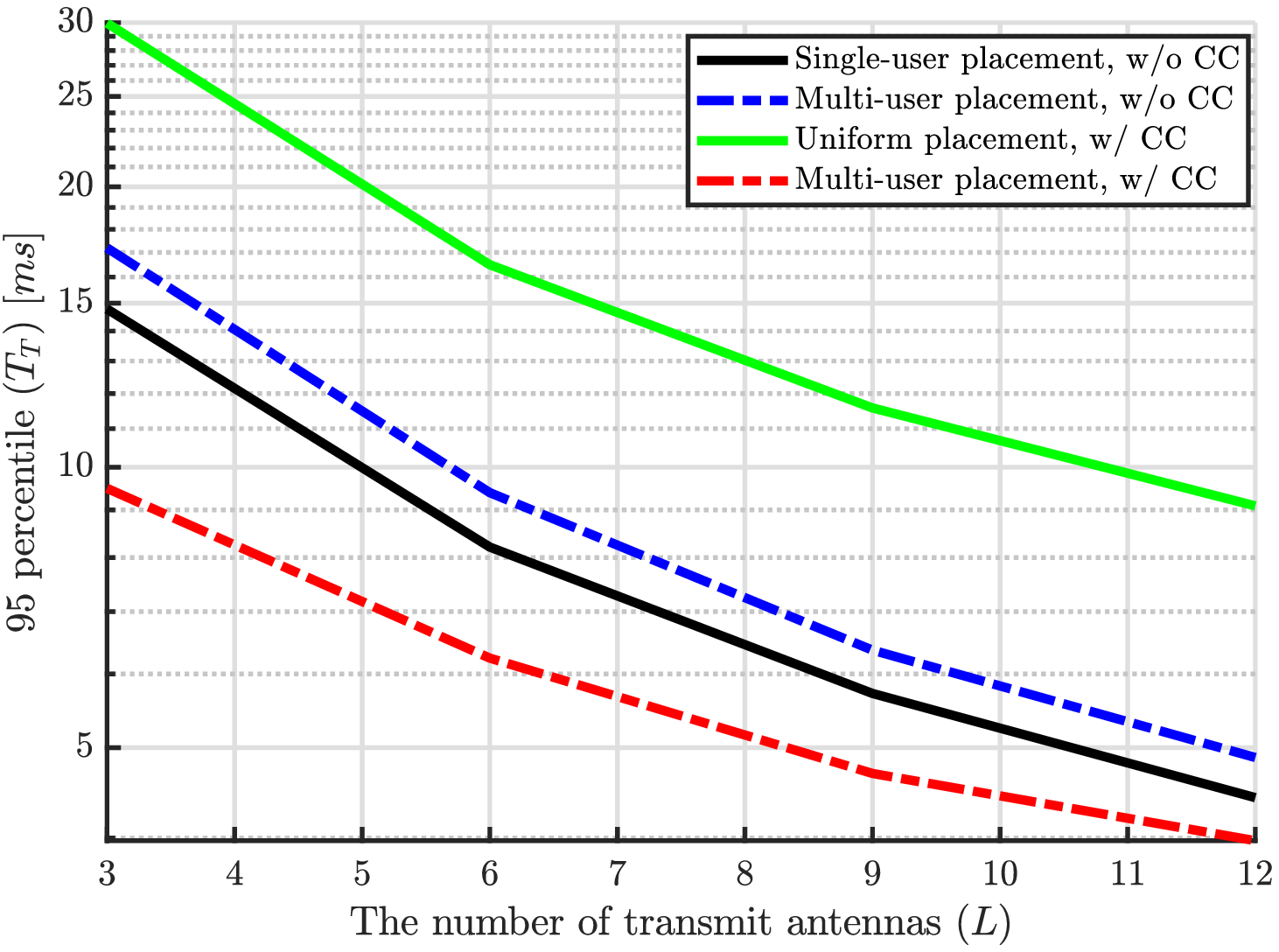}
\caption{Delivery time versus (Logarithmic-scale) the number of antennas ($L$), where $K = 36, M/S = 0.33$, and $\sigma=7$ [dB].}
\label{fig:sim_L_95time}
\end{minipage}
\hspace{1mm}
\begin{minipage}[c]{0.32\textwidth}
\centering 
\setlength\abovecaptionskip{-0.25\baselineskip}
\includegraphics[width=1\columnwidth,keepaspectratio]{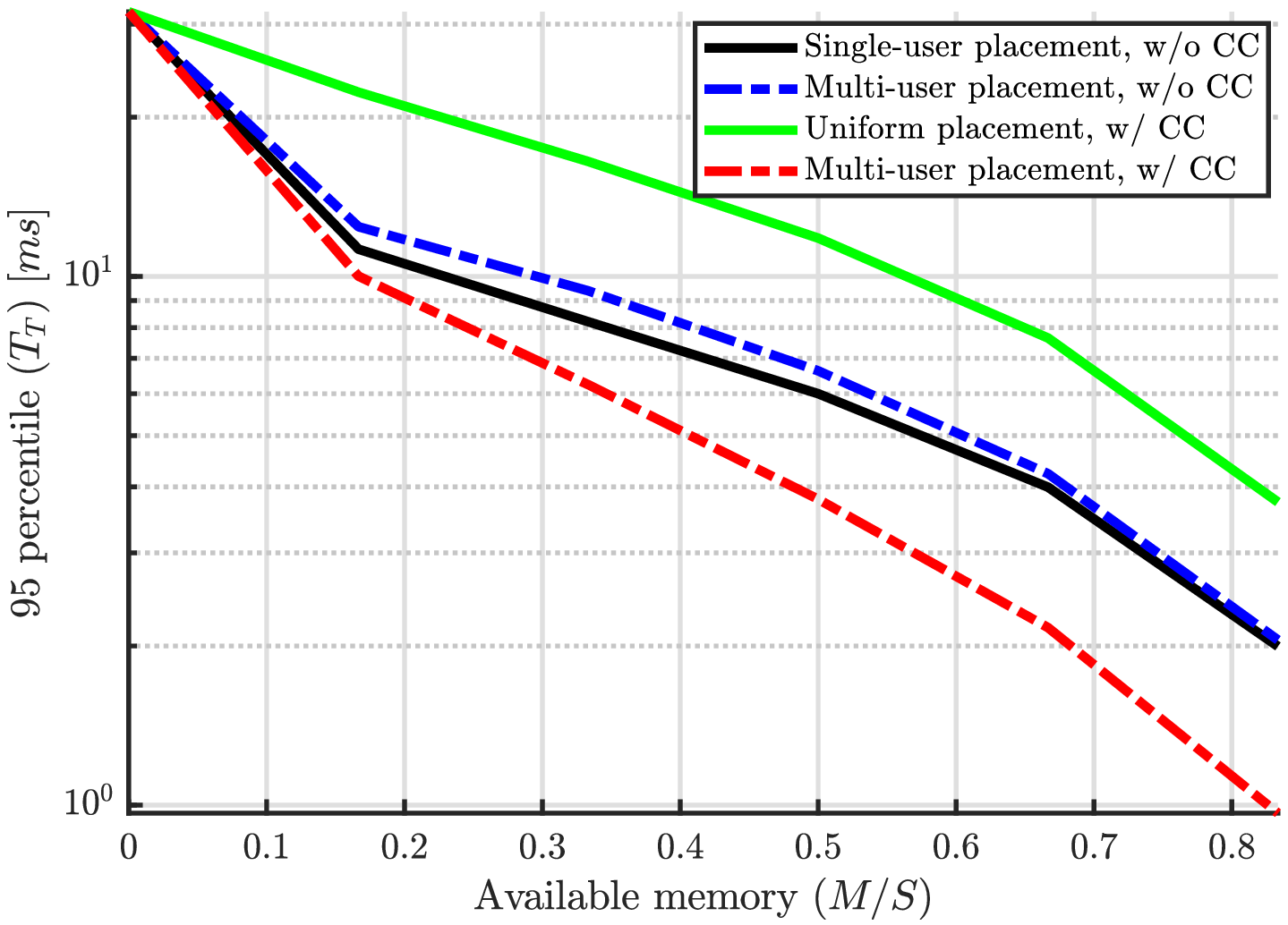}
\caption{Delivery time (logarithmic-scale) versus memory size ($M/S$), where $K = 36, \sigma=7$, and $L=6$.}
\label{fig:sim_M_95time}
\end{minipage}
\hspace{1mm}
\begin{minipage}[c]{0.32\textwidth}
\centering 
\setlength\abovecaptionskip{-0.25\baselineskip}
\includegraphics[width=1\columnwidth,keepaspectratio]{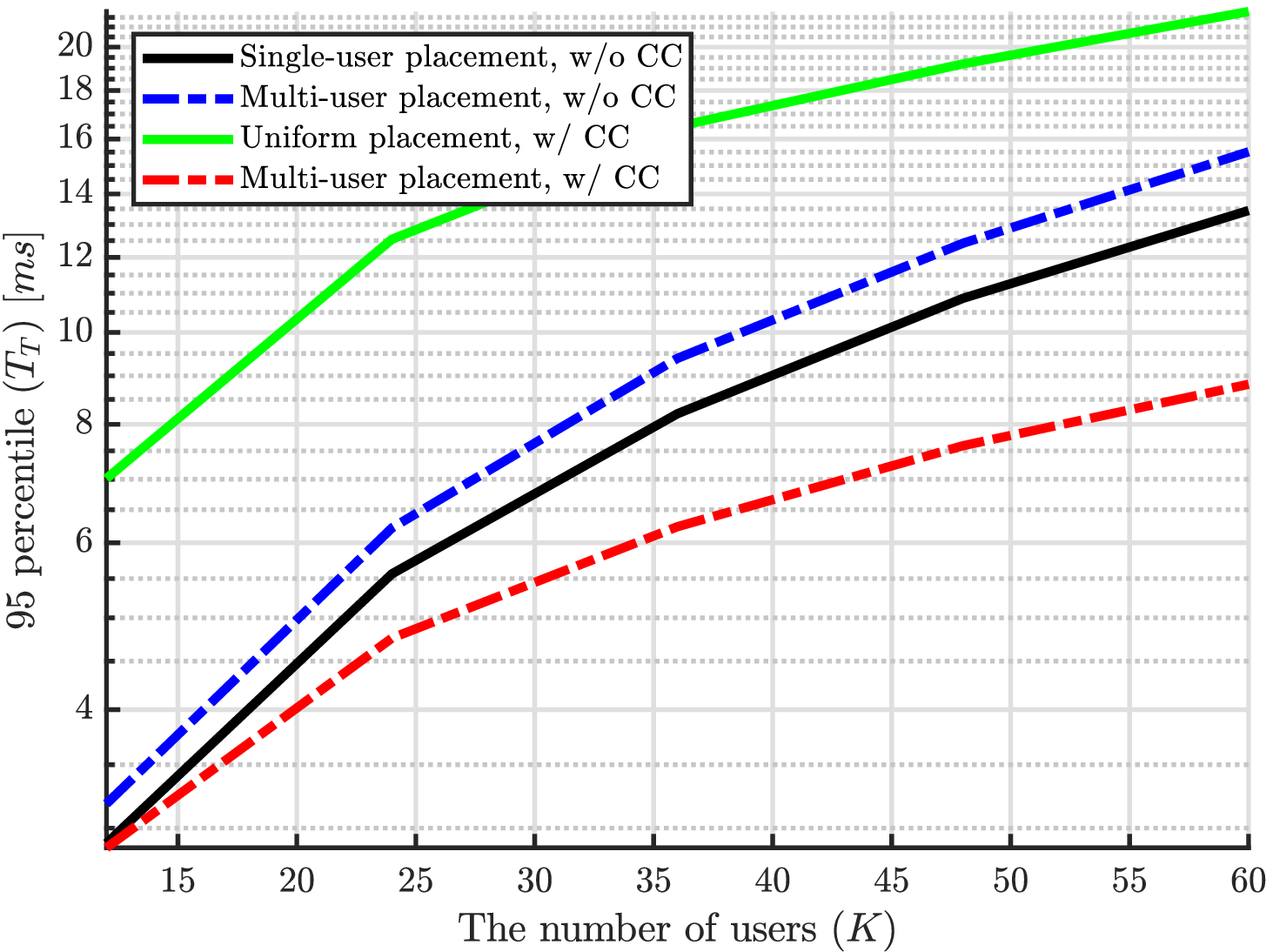}
\caption{Delivery time (logarithmic-scale) versus user count ($K$), where $M/S = 0.33, \sigma=7$, and $L=6$.}
\label{fig:sim_K_95time}
\end{minipage}
\end{figure*}

\section{Conclusion}
\label{sec:conclusions}
In this paper, we have proposed a cache placement and delivery scheme for location-dependent data requests suitable for future collaborative wireless XR applications. Our scheme mitigates excessive delivery times through an efficient memory allocation process. We allocate a large portion of memory to bottleneck areas by approximating a rate difference between various locations in the application hall. Due to its simple transceiver design, reduced subpacketization and processing requirements, and the ability to employ shared caching concepts, the proposed scheme can be easily implemented in large networks with many users utilizing a well-defined LAPDA structure. Numerical results demonstrated the superiority of the proposed approach in different scenarios, especially those with high channel variations and a large number of users, for which a bounded transmission time was ensured by minimizing the use of wireless resources in bottleneck areas. In the future, the proposed approach could be extended to include multiple transmitters, incorporate side information regarding user movements and STU transition probabilities, and examine more dynamic scenarios where users' cache content is updated as they navigate through the environment.

\appendices
\section{FM Matrix Formation} \label{sec: Appendix A}
Section~\ref{sec:delivery} introduces the mapping process from RIS $\CN_k$ to the TIS $\Check{\CN}_{k}$ to serve user $k$ appropriately. To facilitate this mapping, $K$ MF matrices $\BG_k$ are defined. The element $g_{i,j,k}$ in $\BG_k$ specifies the number of file fragments $\{W^{q}_{j}(s_k), j \in \CN_k, q \in [D_k]\}$  that should be substituted for the temporary index $i \in \Check{\CN}_k$. According to Algorithm~\ref{Alg:Delivery_main} for the delivery process, the content transmitted with temporary index $i \in \Check{\CN}_k$ during transmission $n \in [\hat{N}]$ must be cached by all users $k \in \overline{\CU}_i(n)$, where $\overline{\CU}_i(n) := \{k \ | \ k \in \CU(n), \hat{\BQ}(i,k) = *\}$. On the other hand, based on the cache placement described in section~\ref{sec:cache_placement}, sub-file $W_{j}(s)$ is cached by all users $k \in \CU_j(s)$, where $\CU_j(s) := \{k \ | \ k \in [K], \ \BQ_{s}(j,k) = *\}$. Hnece, the value of $g_{i,j,k}$ is determined as follows:
\begin{equation}
\begin{aligned}
    \begin{cases} \label{cond: feasibility}\nonumber
    g_{i,j,k} \geq 0,& \text{if } \overline{\CU}_i(n) \subseteq \CU_j (s_{k})  \\
   g_{i,j,k} = 0,              & \text{otherwise}
\end{cases}, \quad \Large{\substack{\forall k \in [K], \forall i \in \Check{\CN}_k , \\ \forall j \in \CN_k }} \; . 
\end{aligned}
\end{equation}
In cases where $\overline{\CU}_i(n) \nsubseteq \CU_j(s_k)$ for some $j \in \CN_k$ and $i \in \Check{\CN}_k$, it means that there is at least one user in the set $\overline{\CU}_i(n)$ who does not have $W^{q}{j}(s_k)$ in its cache. This violation of condition $C.4$ in Definition~\ref{MLPDA_definition_new} results in interference-limited transmission. Consequently, for such cases, no file fragment ${W^{q}_{j}(s_k)}$ should substituted for the delivery index $i \in \Check{\CN}_k$, i.e., $g_{i,j,k} = 0$.

Next, as discussed in section~\ref{sec:delivery}, the total number of file-fragments substituted for a temporary index $i \in \Check{\CN}_k$ must be equal to $\alpha$. Thus, the non-zero variables $g_{i,j,k}$ must satisfy the following
\begin{equation} \label{cond: total-delivery-chunk}
    \sum_{j \in \CN_k}g_{i,j,k} = \alpha, \quad  \forall k \in [K], \forall i \in \Check{\CN}_k. 
\end{equation}
Additionally, the total number of file fragments for each sub-file $W_j(s_k)$ is equal to $D_k$. Therefore, the total number of file fragments $\{W^{q}_{j}(s_k), j \in \CN_k\}$ substituted for different temporary indexes $\{i \, | \, \overline{\CU}_i(n) \subseteq \CU_j (s_{k}) , i \in \Check{\CN}_k\}$ should sum up to $D_k$. In other words,
\begin{equation} \label{cond: total-file-chunk}
    \sum_{i \in \Check{\CN}_k}g_{i,j,k} = D_k, \quad  \forall k \in [K], \forall j \in \CN_k.
\end{equation}
Consequently, to form each user-specific matrix $\BG_k$, the server needs to solve the non-zero variables $\{g_{i,j,k}\}$, satisfying $\hat{F}-\hat{Z}$ conditions in~\eqref{cond: total-delivery-chunk} and $F_{s_k}-Z_{s_k}$ conditions in~\eqref{cond: total-file-chunk}. Hence, to form the matrix $\BG_k$, the server needs to solve the following system of equations
\begin{equation}\label{eq: FM-equation}
    \BA_k\By = \Bb_k,
\end{equation}
where, $\BA_k \in \{0,1\}^{\varphi_k \times \psi_k}$ is the coefficient matrix, $\varphi_k = \hat{F}+F_{s_k}-(\hat{Z}+Z_{s_k})$ is the total number of conditions in~\eqref{cond: total-delivery-chunk} and~\eqref{cond: total-file-chunk}, $\psi_k$ is the total number of non-zero variables, $\By \in \mathbb{W}^{\psi_k}$ is the variable vector, and $\Bb_k \in \{\alpha , D_k\}^{\varphi_k}$ is the target vector. 
%$\BA_k \in \{0,1\}^{\big(2K-\hat{t}-t_k\big)\times\big((t_k-\hat{t}+1)(K-t_k)\big)}$ is the coefficient matrix, $\Bx_k \in \mathbb{N}^{\big((t_k-\hat{t}+1)(K-t_k)\big)}$ is the variable vector, and $\Bb_k \in \{\alpha(\hat{t}+L), \, D_k\}^{\big(2K-\hat{t}-t_k\big)}$ is the target vector.
Based on each user state $s_k$, variable $\By$ in~\eqref{eq: FM-equation} is given in one of the following closed-form~\cite{Linear_Algebra}
\begin{enumerate}
    \item Over-determined case \big($Rank(\BA_k) = \varphi_k$\big) :
    \begin{equation} \nonumber
        \By = (\BA_k^{T}\BA_k)^{-1}\BA_k^{T}\Bb_k,
    \end{equation}
\item Under-determined case \big($Rank(\BA_k) = \psi_k$\big) :
    \begin{equation}\nonumber
        \By = \BA_k^{T}(\BA_k\BA_k^{T})^{-1}\Bb_k + (\BI - \BA_k^{T}(\BA_k\BA_k^{T})^{-1}\BA_k)\By_0,
    \end{equation} 
\end{enumerate}
In case $Rank(\BA_k) < \min\{\varphi_k, \psi_k\}$, there are several methods (e.g., singular value decomposition and rank decomposition~\cite{Linear_Algebra}) available to solve~\eqref{eq: FM-equation}, which are beyond the scope of this study. It is worth noting that Definition~\ref{LAPDA_definition} ensures that there is at least one $g_{i,j,k}>0$ for every $i \in \Check{\CN}_k$ in equation~\eqref{cond: total-delivery-chunk}, and similarly for every $j \in \CN_k$ in equation~\eqref{cond: total-file-chunk}. Consequently, these equalities remain valid at all times, and equation~\eqref{eq: FM-equation} has a solution that is not empty.  %Interested readers are referred to~\cite{Linear_Algebra}.

\section{Non-integer Coded Caching gains}\label{sec:Appendix_non_int}
We follow a similar memory-sharing scheme as in~\cite{MaddahAli-2014} for non-integer \emph{coded caching} gains (i.e., in case $Km(s)$ is not  an integer). In this regard:
\begin{itemize}
    \item[1.] File $W(s)$ of state $s$ is first divided into two non-overlapping parts $W_1(s)$ and $W_2(s)$, where $|W_1(s)| = \left(\lfloor Km(s) \rfloor +1 - Km(s)\right)|W(s)|$ and $|W_2(s)| = \left(Km(s)-\lfloor Km(s) \rfloor\right)|W(s)|$.
    \item[2.] Two separate MLPDAs $\overline{\BQ}(s)$ and $\underline{\BQ}(s)$ are formed based on Definition~\ref{MLPDA_definition_new}, such that $\frac{\overline{Z}(s)}{\overline{F}(s)} = \overline{m}(s) = \frac{\lfloor Km(s) \rfloor + 1}{K}$ and $\frac{\underline{Z}(s)}{\underline{F}(s)} = \underline{m}(s)=\frac{\lfloor Km(s) \rfloor}{K}$.
    \item[3.] Each user caches the data part $W_1(S)$ based on $\underline{\BQ}(s)$, and data part $W_2(S)$ based on $\overline{\BQ}(s)$, according to the placement scheme in section~\ref{sec:cache_placement}.
\end{itemize}
It can be verified that the proposed memory-sharing process does not violate the cache constraint, i.e., $\frac{\lfloor Km(s) \rfloor}{K}\big(\lfloor Km(s) \rfloor +1 - Km(s)\big) + \frac{\lfloor Km(s) \rfloor+1}{K}\big(\lfloor Km(s) - Km(s) \rfloor\big)= \frac{Km(s)}{K} = m(s)$. 

In this case, the delivery will be done in two sub-phases based on time-sharing. In the first sub-phase, $\beta_1 = \big(\lfloor Km(s_{k^{*}}) \rfloor +1 - Km(s_{k^{*}})\big)$ portion of the files are delivered based on $\hat{\BQ}_1 \equiv \underline{\BQ}(s_{k^{*}})$, where $k^{*} = \min_{k \in \CK} m_k$. The remaining  $\beta_2 = \left(Km(s_{k^{*}})-\lfloor Km(s_{k^{*}}) \rfloor\right)$ portion of files are delivered in the second sub-phase based on $\hat{\BQ}_2 \equiv \overline{\BQ}(s_{k^{*}})$. We consider two cases for data delivery, 1) $\lfloor Km(s_k) \rfloor > \lfloor Km(s_{k^{*}}) \rfloor$ and 2) $\lfloor Km(s_k) \rfloor = \lfloor Km(s_{k^{*}}) \rfloor$.

\noindent\textbf{\textbf{Case 1 ($\mathbf{\lfloor Km(s_k) \rfloor > \lfloor Km(s_{k^{*}}}) \rfloor$):}} Since in this case, for all user $k$, the placement is done differently for $W_1(s_k)$ and $W_2(s_k)$, the index mapping process is also separately performed for $W_1(s_k)$ and $W_2(s_k)$. To this end, during the first delivery sub-phase, $\beta_1$ portion of every \textit{sub-file} of $W_1(s_k)$ is divided into $\underline{D}^1_k = \underline{\alpha}\frac{\underline{F}-\underline{Z}}{\underline{F}_{s_k}-\underline{Z}_{s_k}}$ smaller fragments, where $\underline{F}$ and $\underline{Z}$ are equivalent to $\underline{F}_{s_{k^{*}}}$ and $\underline{Z}_{s_{k^{*}}}$, respectively. Similarly, $\beta_1$ portion of every \textit{sub-file} of $W_2(s_k)$ is divided into  and $\underline{D}^2_k = \underline{\alpha}\frac{\underline{F}-\underline{Z}}{\overline{F}_{s_k}-\overline{Z}_{s_k}}$ smaller fragments. Then, based on $\hat{\BQ}_1, \underline{\BQ}(s_{k})$, and $\overline{\BQ}(s_{k})$, two file-fragment matrices $\underline{\BP}_k^1$ and $\underline{\BP}_k^2$ are formed for $W_1(s_k)$ and $W_2(s_k)$, respectively. Using $\underline{\BP}_k^1$ and $\underline{\BP}_k^2$, each transmitted message to user $k$ will carry $\underline{\alpha}$ file-fragments of $\beta_1$ portion of $W_1(s_k)$ and $\underline{\alpha}$ file-fragments of $\beta_1$ portion of $W_2(s_k)$. The remaining $\beta_2$ portion of $W_1(s_k)$ and $W_2(s_k)$ will be delivered in the second sub-phase, using $\hat{\BQ}_2, \underline{\BQ}(s_{k})$, and $\overline{\BQ}(s_{k})$ to form $\overline{\BP}_k^1$ and $\overline{\BP}_k^2$. For the sake of brevity, we avoid reviewing a similar process in the second delivery sub-phase.

\noindent\textbf{\textbf{Case 2 ($\mathbf{\lfloor Km(s_k) \rfloor = \lfloor Km(s_{k^{*}}}) \rfloor$):}} In this case, $\big(\lfloor Km(s_k) \rfloor +1 - Km(s_k)\big)$ portion of $W(s_k)$ (i.e., $W_1(s_k)$) is already cached based on $\hat{\BQ}_1$ and can be easily delivered by forming $\underline{\BP}_k^1$. The remaining $\beta_1 - \big(\lfloor Km(s_k) \rfloor +1 - Km(s_k)\big)$ portion of $W(s_k)$ can also be delivered in the first delivery sub-phase based on $\hat{\BQ}_1$ and $\overline{\BQ}(s_{k})$. In this regard, first $\beta_1 - \big(\lfloor Km(s_k) \rfloor +1 - Km(s_k)\big)$ portion of of every \textit{sub-file} of $W_2(s_k)$ is divided into and $\underline{D}^2_k = \underline{\alpha}\frac{\underline{F}-\underline{Z}}{\overline{F}_{s_k}-\overline{Z}_{s_k}}$ smaller fragments. Then, using $\underline{\BP}_k^1$ and $\underline{\BP}_k^2$, each transmitted message to user $k$ will carry $\underline{\alpha}$ file-fragments of $W_1(s_k)$ and $\underline{\alpha}$ file-fragments of $\beta_1 - \big(\lfloor Km(s_k) \rfloor +1 - Km(s_k)\big)$ portion of $W_2(s_k)$. In the second sub-phase, the remaining portion of $W(s_k)$, i.e., $\Big[\big(Km(s_k)-\lfloor Km(s_k) \rfloor\big)-\Big(\beta_1 - \big(\lfloor Km(s_k) \rfloor +1 - Km(s_k)\big)\Big)\Big] = 1-\beta_1 = \beta_2$ portion of $W_2(s_k)$, will be delivered based on $\hat{\BQ}_2$ and $\overline{\BQ}(s_{k})$.


\begin{thebibliography}{10}
\providecommand{\url}[1]{#1}
\csname url@samestyle\endcsname
\providecommand{\newblock}{\relax}
\providecommand{\bibinfo}[2]{#2}
\providecommand{\BIBentrySTDinterwordspacing}{\spaceskip=0pt\relax}
\providecommand{\BIBentryALTinterwordstretchfactor}{4}
\providecommand{\BIBentryALTinterwordspacing}{\spaceskip=\fontdimen2\font plus
\BIBentryALTinterwordstretchfactor\fontdimen3\font minus
  \fontdimen4\font\relax}
\providecommand{\BIBforeignlanguage}[2]{{%
\expandafter\ifx\csname l@#1\endcsname\relax
\typeout{** WARNING: IEEEtran.bst: No hyphenation pattern has been}%
\typeout{** loaded for the language `#1'. Using the pattern for}%
\typeout{** the default language instead.}%
\else
\language=\csname l@#1\endcsname
\fi
#2}}
\providecommand{\BIBdecl}{\relax}
\BIBdecl

\bibitem{cisco2020}
Cisco, ``{Cisco Annual Internet Report, 2018--2023},'' \emph{White Paper},
  vol.~1, march, 2020.

\bibitem{salehi2022enhancing}
M.~Salehi, K.~Hooli, J.~Hulkkonen, and A.~T{\"o}lli, ``Enhancing
  next-generation extended reality applications with coded caching,''
  \emph{IEEE Open Journal of the Communications Society}, 2023.

\bibitem{flashback_2016_VR_static_dynamic_support}
K.~Boos, D.~Chu, and E.~Cuervo, ``{Flashback: Immersive virtual reality on
  mobile devices via rendering memoization},'' in \emph{Proceedings of the 14th
  Annual International Conference on Mobile Systems, Applications, and
  Services}, 2016, pp. 291--304.

\bibitem{Nokia-immersive}
\BIBentryALTinterwordspacing
E.~Thomas, E.~Potetsianakis, T.~Stockhammer, I.~Bouazizi, and M.-L. Champel,
  ``{MPEG Media Enablers For Richer XR Experiences},'' \emph{arXiv}, 2020.
  [Online]. Available: \url{https://arxiv.org/abs/2010.04645.}
\BIBentrySTDinterwordspacing

\bibitem{6G_white_paper_2020}
N.~Rajatheva, I.~Atzeni, E.~Bjornson, A.~Bourdoux, S.~Buzzi, J.-B. Dore,
  S.~Erkucuk, M.~Fuentes, K.~Guan, Y.~Hu \emph{et~al.}, ``White paper on
  broadband connectivity in {6G},'' \emph{arXiv preprint arXiv:2004.14247},
  2020.

\bibitem{bastug2017toward}
E.~Bastug, M.~Bennis, M.~M{\'{e}}dard, and M.~Debbah, ``{Toward interconnected
  virtual reality: Opportunities, challenges, and enablers},'' \emph{IEEE
  Communications Magazine}, vol.~55, no.~6, pp. 110--117, 2017.

\bibitem{taleb2022_towards_XR_vision}
T.~Taleb, A.~Boudi, L.~Rosa, L.~Cordeiro, T.~Theodoropoulos, K.~Tserpes,
  P.~Dazzi, A.~Protopsaltis, and R.~Li, ``{Towards supporting {XR} services:
  Architecture and enablers},'' \emph{IEEE Internet of Things Journal}, 2022.

\bibitem{walid_sad_bennis_VR_XR_2022}
C.~Chaccour, M.~N. Soorki, W.~Saad, M.~Bennis, and P.~Popovski, ``Can terahertz
  provide high-rate reliable low latency communications for wireless vr?''
  \emph{IEEE Internet of Things Journal}, 2022.

\bibitem{wireless_virtual_reality_TCOM2018}
M.~Chen, W.~Saad, and C.~Yin, ``Virtual reality over wireless networks:
  Quality-of-service model and learning-based resource management,'' \emph{IEEE
  Transactions on Communications}, vol.~66, no.~11, pp. 5621--5635, 2018.

\bibitem{URLLC_ICC_2017}
G.~Pocovi, B.~Soret, K.~I. Pedersen, and P.~Mogensen, ``{MAC} layer
  enhancements for ultra-reliable low-latency communications in cellular
  networks,'' in \emph{2017 IEEE International Conference on Communications
  Workshops (ICC Workshops)}.\hskip 1em plus 0.5em minus 0.4em\relax IEEE,
  2017, pp. 1005--1010.

\bibitem{proactive_caching_Cm2016}
H.~Liu, Z.~Chen, and L.~Qian, ``The three primary colors of mobile systems,''
  \emph{IEEE Comm. Mag.}, vol.~54, no.~9, pp. 15--21, 2016.

\bibitem{role_of_caching_in_future_wireless_caire_JSAC2018}
G.~S. Paschos, G.~Iosifidis, M.~Tao, D.~Towsley, and G.~Caire, ``The role of
  caching in future communication systems and networks,'' \emph{IEEE Journal on
  Selected Areas in Communications}, vol.~36, no.~6, pp. 1111--1125, 2018.

\bibitem{sun2019communications}
Y.~Sun, Z.~Chen, M.~Tao, and H.~Liu, ``{Communications, caching, and computing
  for mobile virtual reality: Modeling and tradeoff},'' \emph{IEEE Transactions
  on Communications}, vol.~67, no.~11, pp. 7573--7586, 2019.

\bibitem{yang2018communication}
X.~Yang, Z.~Chen, K.~Li, Y.~Sun, N.~Liu, W.~Xie, and Y.~Zhao,
  ``Communication-constrained mobile edge computing systems for wireless
  virtual reality: Scheduling and tradeoff,'' \emph{IEEE Access}, vol.~6, pp.
  16\,665--16\,677, 2018.

\bibitem{sun2020bandwidth}
Y.~Sun, Z.~Chen, M.~Tao, and H.~Liu, ``Bandwidth gain from mobile edge
  computing and caching in wireless multicast systems,'' \emph{IEEE Trans. on
  Wireless Comm.}, vol.~19, no.~6, pp. 3992--4007, 2020.

\bibitem{dang2019joint}
T.~Dang and M.~Peng, ``Joint radio communication, caching, and computing design
  for mobile virtual reality delivery in fog radio access networks,''
  \emph{IEEE Journal on Selected Areas in Communications}, vol.~37, no.~7, pp.
  1594--1607, 2019.

\bibitem{bastug2014living}
E.~Bastug, M.~Bennis, and M.~Debbah, ``{Living on the edge: The role of
  proactive caching in 5G wireless networks},'' \emph{IEEE Communications
  Magazine}, vol.~52, no.~8, pp. 82--89, 2014.

\bibitem{yang2015analysis}
C.~Yang, Y.~Yao, Z.~Chen, and B.~Xia, ``Analysis on cache-enabled wireless
  heterogeneous networks,'' \emph{IEEE Transactions on Wireless
  Communications}, vol.~15, no.~1, pp. 131--145, 2015.

\bibitem{MaddahAli-2014}
M.~A. Maddah-Ali and U.~Niesen, ``Fundamental limits of caching,'' \emph{{IEEE}
  Trans. Inform. Theory}, vol.~60, no.~5, pp. 2856--2867, May 2014.

\bibitem{CC_edge_computing_for_VR_twc2021}
Y.~Li, Z.~Chen, and M.~Tao, ``Coded caching with device computing in mobile
  edge computing systems,'' \emph{IEEE Transactions on Wireless
  Communications}, 2021.

\bibitem{Mahmoodi_immersive_isit2021}
H.~B. Mahmoodi, M.~Salehi, and A.~Tölli, ``Non-symmetric coded caching for
  location-dependent content delivery,'' in \emph{2021 IEEE International
  Symposium on Information Theory (ISIT)}, 2021, pp. 712--717.

\bibitem{mahmoodi2022asymmetric}
H.~B. Mahmoodi, M.~J. Salehi, and A.~T{\"o}lli, ``Asymmetric multi-antenna
  coded caching for location-dependent content delivery,'' in \emph{GLOBECOM
  2022-2022 IEEE Glob. Comm. Conf.}, 2022, pp. 1930--1935.

\bibitem{mahmoodi_ICC2022_nonsym}
H.~B. Mahmoodi, M.~Salehi, and A.~T{\"o}lli, ``Non-symmetric multi-antenna
  coded caching for location-dependent content delivery,'' in \emph{ICC
  2022-IEEE International Conference on Communications}.\hskip 1em plus 0.5em
  minus 0.4em\relax IEEE, 2022, pp. 5165--5170.

\bibitem{mahmoodi2022asymmetric_arxiv_TWC2023}
H.~B. Mahmoodi, M.~Salehi, and A.~Tölli, ``Multi-antenna coded caching for
  location-dependent content delivery,'' \emph{IEEE Transactions on Wireless
  Communications}, pp. 1--1, 2023.

\bibitem{pooya-cc-physical-2019-journal}
S.~P. {Shariatpanahi}, G.~{Caire}, and B.~{Hossein Khalaj}, ``Physical-layer
  schemes for wireless coded caching,'' \emph{{IEEE} Trans. Inform. Theory},
  vol.~65, no.~5, pp. 2792--2807, 2019.

\bibitem{Shariatpanahi2016}
S.~P. Shariatpanahi \emph{et~al.}, ``Multi-server coded caching,'' \emph{{IEEE}
  Trans. Inform. Theory}, vol.~62, no.~12, pp. 7253--7271, Dec 2016.

\bibitem{yu2017exact}
Q.~Yu, M.~A. Maddah-Ali, and A.~S. Avestimehr, ``{The exact rate-memory
  tradeoff for caching with uncoded prefetching},'' \emph{IEEE International
  Symposium on Information Theory - Proceedings}, vol.~64, no.~2, pp.
  1613--1617, 2017.

\bibitem{tolli2017multi}
A.~T{\"{o}}lli, S.~P. Shariatpanahi, J.~Kaleva, and B.~H. Khalaj,
  ``{Multi-antenna interference management for coded caching},'' \emph{IEEE
  Transactions on Wireless Communications}, vol.~19, no.~3, pp. 2091--2106,
  2020.

\bibitem{Ji2016}
M.~Ji, G.~Caire, and A.~F. Molisch, ``Fundamental limits of caching in wireless
  {D2D} networks,'' \emph{{IEEE} Trans. Inform. Theory}, vol.~62, no.~2, pp.
  849--869, Feb 2016.

\bibitem{D2D-CC-Optload-memtradeof-caire-2019}
C.~{Yapar}, K.~{Wan}, R.~F. {Schaefer}, and G.~{Caire}, ``On the optimality of
  {D2D} coded caching with uncoded cache placement and one-shot delivery,''
  \emph{{IEEE} Trans. Commun.}, vol.~67, no.~12, pp. 8179--8192, 2019.

\bibitem{D2D_CC_mahmoodi_2022}
H.~B. Mahmoodi, J.~Kaleva, S.~P. Shariatpanahi, and A.~Tölli, ``{D2D} assisted
  multi-antenna coded caching,'' \emph{IEEE Access}, vol.~11, pp.
  16\,271--16\,287, 2023.

\bibitem{PDA_first_2017}
Q.~Yan, M.~Cheng, X.~Tang, and Q.~Chen, ``On the placement delivery array
  design for centralized coded caching scheme,'' \emph{IEEE Transactions on
  Information Theory}, vol.~63, no.~9, pp. 5821--5833, 2017.

\bibitem{cheng_generalized_PDA_2019}
M.~Cheng, J.~Jiang, Q.~Wang, and Y.~Yao, ``A generalized grouping scheme in
  coded caching,'' \emph{IEEE Transactions on Communications}, vol.~67, no.~5,
  pp. 3422--3430, 2019.

\bibitem{DPDA2019}
J.~{Wang}, M.~{Cheng}, Q.~{Yan}, and X.~{Tang}, ``Placement delivery array
  design for coded caching scheme in {D2D} networks,'' \emph{{IEEE} Trans.
  Commun.}, vol.~67, no.~5, pp. 3388--3395, May 2019.

\bibitem{cheng_centralized_PDA_framework_2021}
M.~Cheng, J.~Wang, X.~Zhong, and Q.~Wang, ``A framework of constructing
  placement delivery arrays for centralized coded caching,'' \emph{IEEE
  Transactions on Information Theory}, vol.~67, no.~11, pp. 7121--7131, 2021.

\bibitem{lampiris2018adding}
E.~Lampiris and P.~Elia, ``{Adding transmitters dramatically boosts
  coded-caching gains for finite file sizes},'' \emph{IEEE Journal on Selected
  Areas in Communications}, vol.~36, no.~6, pp. 1176--1188, 2018.

\bibitem{Shahred_cache_Emanuele_2020}
E.~Parrinello, A.~Ünsal, and P.~Elia, ``Fundamental limits of coded caching
  with multiple antennas, shared caches and uncoded prefetching,'' \emph{IEEE
  Transactions on Information Theory}, vol.~66, no.~4, pp. 2252--2268, 2020.

\bibitem{Shared_cache_decentrlized_2021}
M.~Dutta and A.~Thomas, ``Decentralized coded caching for shared caches,''
  \emph{IEEE Comm. Letters}, vol.~25, no.~5, pp. 1458--1462, 2021.

\bibitem{salehi2021MIMO}
M.~J. Salehi, H.~B. Mahmoodi, and A.~T{\"{o}}lli, ``{A Low-Subpacketization
  High-Performance MIMO Coded Caching Scheme},'' in \emph{WSA 2021 - 25th
  International ITG Workshop on Smart Antennas}, 2021, pp. 427--432.

\bibitem{salehi2023multicast}
M.~Salehi, M.~NaseriTehrani, and A.~T{\"o}lli, ``Multicast beamformer design
  for {MIMO} coded caching systems,'' in \emph{ICASSP 2023-2023 IEEE
  International Conference on Acoustics, Speech and Signal Processing
  (ICASSP)}.\hskip 1em plus 0.5em minus 0.4em\relax IEEE, 2023, pp. 1--5.

\bibitem{abolpour2022coded}
M.~Abolpour, M.~J. Salehi, and A.~Tolli, ``{Coded Caching and Spatial
  Multiplexing Gain Trade-off in Dynamic MISO Networks},'' \emph{IEEE Workshop
  on Signal Processing Advances in Wireless Communications, SPAWC}, vol.
  2022-July, 2022.

\bibitem{abolpour2023cache}
M.~Abolpour, M.~Salehi, and A.~T{\"o}lli, ``Cache-aided communications in
  {MISO} networks with dynamic user behavior: A universal solution,'' in
  \emph{2023 IEEE Int. Symp. on Inf. Theory (ISIT), available
  at:arXiv:2304.11623}, 2023.

\bibitem{salehi2021lowsubpacketization}
M.~Salehi, E.~Parrinello, H.~B. Mahmoodi, and A.~T{\"o}lli,
  ``Low-subpacketization multi-antenna coded caching for dynamic networks,'' in
  \emph{2022 Joint European Conference on Networks and Communications \& 6G
  Summit (EuCNC/6G Summit)}.\hskip 1em plus 0.5em minus 0.4em\relax IEEE, 2022,
  pp. 112--117.

\bibitem{MLPDA_ISIT2022}
T.~Yang, K.~Wan, M.~Cheng, R.~C. Qiu, and G.~Caire, ``Multiple-antenna
  placement delivery array for cache-aided miso systems,'' \emph{IEEE
  Transactions on Information Theory}, 2023.

\bibitem{salehi2019subpacketization}
M.~Salehi, A.~Tolli, S.~P. Shariatpanahi, and J.~Kaleva,
  ``{Subpacketization-rate trade-off in multi-antenna coded caching},'' in
  \emph{2019 IEEE Global Communications Conference, GLOBECOM 2019 -
  Proceedings}.\hskip 1em plus 0.5em minus 0.4em\relax IEEE, 2019, pp. 1--6.

\bibitem{salehi2022multi}
M.~Salehi and A.~T{\"{o}}lli, ``{Multi-antenna Coded Caching at Finite-SNR:
  Breaking Down the Gain Structure},'' \emph{arXiv preprint arXiv:2210.10433},
  2022.

\bibitem{zhao_petros_MU_MISO_near_far_issue_ITW2021}
H.~Zhao, A.~Bazco-Nogueras, and P.~Elia, ``Resolving the worst-user bottleneck
  of coded caching: Exploiting finite file sizes,'' in \emph{2020 IEEE
  Information Theory Workshop (ITW)}.\hskip 1em plus 0.5em minus 0.4em\relax
  IEEE, 2021, pp. 1--5.

\bibitem{destounis2020adaptive}
A.~Destounis, A.~Ghorbel, G.~S. Paschos, and M.~Kobayashi, ``{Adaptive Coded
  Caching for Fair Delivery over Fading Channels},'' \emph{IEEE Transactions on
  Information Theory}, 2020.

\bibitem{Coded_caching_for_stochastic_wireless_network_TWC2021}
Y.~Gu, C.~Yang, B.~Xia, and D.~Xu, ``Design and analysis of coded caching
  schemes in stochastic wireless networks,'' \emph{IEEE Transactions on
  Wireless Communications}, 2021.

\bibitem{liu_joint_power_energi_cc_TWC2021}
Y.~Liu, A.~Tang, and X.~Wang, ``Joint scheduling and power optimization for
  delay constrained transmissions in coded caching over wireless fading
  channels,'' \emph{IEEE Transactions on Wireless Communications}, 2021.

\bibitem{salehi2020coded}
M.~Salehi, A.~Tolli, and S.~P. Shariatpanahi, ``Coded caching with uneven
  channels: A quality of experience approach,'' in \emph{2020 IEEE Int.
  Workshop on Sig. Proc. Advances in Wireless Comm. (SPAWC)}, 2020, pp. 1--5.

\bibitem{maddah2014fundamental}
M.~A. Maddah-Ali and U.~Niesen, ``{Fundamental limits of caching},'' \emph{IEEE
  Trans. on Inf. Theory}, vol.~60, no.~5, pp. 2856--2867, 2014.

\bibitem{salehi2020lowcomplexity}
M.~Salehi, E.~Parrinello, S.~P. Shariatpanahi, P.~Elia, and A.~Tölli,
  ``Low-complexity high-performance cyclic caching for large {MISO} systems,''
  \emph{IEEE Transactions on Wireless Communications}, vol.~21, no.~5, pp.
  3263--3278, 2022.

\bibitem{tang2017coded}
A.~Tang, S.~Roy, and X.~Wang, ``{Coded caching for wireless backhaul networks
  with unequal link rates},'' \emph{IEEE Transactions on Communications},
  vol.~66, no.~1, pp. 1--13, 2017.

\bibitem{zhao_wireless_CC_ring_area_ISIT2021}
H.~Zhao, A.~Bazco-Nogueras, and P.~Elia, ``Wireless coded caching with shared
  caches can overcome the near-far bottleneck,'' in \emph{2021 IEEE Int. Symp.
  on Inf. Theory (ISIT)}, 2021, pp. 350--355.

\bibitem{zhao_petros_near_far_Nakagami_asilomar2021}
------, ``Coded caching gains at low {SNR} over nakagami fading channels,”,''
  in \emph{Proc. 55th Asilomar Conf. Signals, Syst., Comput.(ACSSC)}, 2021.

\bibitem{parrinello2019fundamental}
E.~Parrinello, A.~{\"{U}}nsal, and P.~Elia, ``{Fundamental Limits of Coded
  Caching with Multiple Antennas, Shared Caches and Uncoded Prefetching},''
  \emph{IEEE Transactions on Information Theory}, 2019.

\bibitem{Linear_Algebra}
H.~Anton, \emph{Elementary Linear Algebra}.\hskip 1em plus 0.5em minus
  0.4em\relax New York, USA: John Wiley, 1987.

\end{thebibliography}
\end{document}